\documentclass[11pt,reqno]{article}

\usepackage{amssymb,amsmath,amsthm,mathtools,wasysym,calc,verbatim,enumitem,tikz,pgfplots,hyperref,url,mathrsfs,fullpage,bbm,comment}
\usepackage[noadjust]{cite}
\usepackage{authblk}
\usepackage{comment,fullpage}
\usepackage[normalem]{ulem}
\mathtoolsset{showonlyrefs}
\pgfplotsset{compat=1.18}
\newtheorem{theorem}{Theorem}
\newtheorem{definition}[theorem]{Definition}

\newtheorem{proposition}[theorem]{Proposition}
\newtheorem{lemma}[theorem]{Lemma}
\newtheorem{claim}[theorem]{Claim}
\newtheorem{corollary}[theorem]{Corollary}

\newtheorem{fact}[theorem]{Fact}

\newtheorem*{claim*}{Claim}

\theoremstyle{remark}
\newtheorem*{remark*}{Remark}

\usepackage{thmtools, thm-restate}

\numberwithin{theorem}{section}

\DeclarePairedDelimiter\norm{\lVert}{\rVert}

\renewcommand{\phi}{\varphi}

\renewcommand{\leq}{\le}
\renewcommand{\geq}{\ge}
\newcommand{\one}{\mathbf{1}}

\newcommand{\eps}{\varepsilon}

\newcommand{\cE}{\mathcal E}
\newcommand{\cG}{\mathcal G}

\newcommand{\cD}{\mathcal{D}}

\newcommand{\cL}{\mathcal L}

\usepackage{cancel}

\newcommand{\cS}{\mathcal S}
\newcommand{\cP}{\mathcal P}

\newcommand{\cI}{\mathcal I}

\newcommand{\E}{\mathbb{E}}

\newcommand{\Var}{\operatorname{Var}}
\newcommand{\tr}{\operatorname{tr}}
\def\1{\mathbbm{1}}
\def\bz{{\bf z}}
\newcommand{\Cov}{\operatorname{Cov}}
\newcommand{\Tr}{\operatorname{Tr}}
\newcommand{\diag}{\operatorname{diag}}
\newcommand{\op}{\operatorname{op}}

\usepackage{tikz-cd}
\usetikzlibrary{positioning,arrows.meta,calc}

\newcommand{\sol}[1]{%
  \mkern1mu\overline{\mkern-1mu #1 \mkern-1mu}\mkern1mu}

\renewcommand{\le}{\leqslant}
\renewcommand{\ge}{\geqslant}
\renewcommand{\P}{\mathbb{P}}

\newcommand{\R}{\mathbb R}

\newcommand{\bh}{\mathbf h}

\newcommand{\bX}{\mathbf X}
\newcommand{\bY}{\mathbf Y}

\newcommand{\EE}{\mathbb E}

\newcommand{\PP}{\mathbb P}

\newcommand{\bx}{\mathbf x}

\newcommand{\bg}{\mathbf g}
\newcommand{\bw}{\mathbf w}
\newcommand{\bm}{\mathbf m}

\usepackage{thm-restate}
\usepackage{array}



\title{Rank-1-perturbed trickledown theorems: Mixing time of Glauber dynamics for the Sherrington-Kirkpatrick model up to $\beta\le\tfrac{1}{2}+\eps$}

\author[1]{Mathews Boban}
\author[1]{Anqi Li}
\author[1]{Shayan Oveis Gharan}
 
\affil[1]{University of Washington, \textsf{\{matheb,aqli,shayan\}@cs.washington.edu}}
\date{}

\begin{document}
\maketitle
\begin{abstract}
    We introduce a new family of trickledown theorems, a.k.a., local to global technique to bound the spectral gap of the Glauber dynamics for multi-state spin systems. In this technique instead of upper-bounding the influence matrix of a link of co-dimension 2 by $\lambda I$ (where $\lambda$ is the second eigenvalue of the link), we upper-bound the influence matrix after a carefully chosen rank-1 shift. The rank-1 shift allows for a significantly smaller upper-bound but it comes at the cost of bounding the average loss due to rank-1 perturbations. 

    As an application we use this method to show that the natural Glauber dynamics mixes in polynomial time to generate samples from the Sherrington-Kirkpatrick model for $\beta\leq 1/2+\eps$, for an absolute constant $\eps>0$. At the heart of the proof we manage to bound the loss due to rank-1 perturbations by {\em averaging} over all links of co-dimension 2. 
\end{abstract}
\section{Introduction}
A multi-state spin system is composed of a set $V=\{1,\dots,n\}$ of particles. Each site $i$ is associated with a set of spins $S_v$. Each possible configuration of the system is an element of the set ${\cal S}:=S_1 \times \dots\times S_n$. Let $\mu$ be a probability distribution over all configurations.

Given such a $\mu$, often one is interested in estimating certain properties of $\mu$, for example, its moments or an expectation of a given function of the spins. The MCMC method is often a method of choice in practice: One designs a Markov chain with stationary distribution $\mu$. Then, proves that this chain has an efficient mixing time, namely polynomial in $n$. This in turn allows one to easily generate samples of $\mu$ by running the Markov chain.

The natural Markov chain is the Glauber dynamics: Given a state $\bx\in {\cal S}$, we choose $1\leq i\leq n$ uniformly at random. Then, we forget the spin of the $i$-th particle and re-sample it according to $\mu$ conditioned on $\bx_{-i}$, i.e., for any $x\in S_i,$ we choose $x$ with probability
$$\frac{\mu(x_1,\dots,x_{i-1},x,x_{i+1},\dots,x_n)}{\sum_{x'\in S_i}\mu(x_1,\dots,x_{i-1},x',x_{i+1},\dots,x_n)}$$

One fundamental question is how long it takes for this chain to mix. Let $\frac{1}{n}\cP$ be the transition probability matrix of the Glauber dynamics Markov chain (the factor $1/n$ is for a notational convenience as will be clear later; essentially it implies every vertex is chosen at a unit rate w.r.t. $\cP$) then the mixing time is defined to be
$$ \max_{\bx\in \cS} \min \bigg\{t: \bigg\|1_\bx^\top \bigg(\frac1{n}\cP \bigg)^t  - \mu \bigg\|_{\mathrm TV}\leq \tfrac{1}{4} \bigg\}.$$

To bound the mixing time of $\frac1{n}\cP$ a method of choice is to bound its spectral gap. 
\begin{theorem}\label{thm:spec-gap}
Suppose that for any function $f:{\cal S}\to\R$ satisfying $\E_{\mu} f = 0$, 
$$ \langle \cL f,f\rangle_\mu\geq \eps\langle f,f\rangle_\mu,$$
where $\cL=nI- \cP$ and for two functions $f,g:{\cal S}\to \R$, $\langle f,g\rangle_\mu=\E_{X\sim\mu}[f(X)g(X)].$ Then, the mixing time of $\frac{1}{n}\cP$ is bounded by $$O\Bigg(\frac{n\cdot \log\Big(\frac{1}{\mu_{\min}}\Big)}{\eps}\Bigg)$$
where $\mu_{\min}=\min_x\mu(x)$.
\end{theorem}

Over the last few decades, several techniques have been developed to bound the mixing time of such a chain such as (path) coupling, canonical paths, and spectral independence. Of particular interest to us is a family of local to global theorems that often go by the name of trickledown theorems (or more classically by the name of Dobrushin's condition or path coupling). 

In such techniques given two arbitrary particles $i,j$ one bounds the {\em worst case} interaction between $i,j$ given the status of the rest of the particles. This is classically studied as the worst case influence of $j$ to $i$ in the sense that by changing the spin of particle $j$ how much the distribution of $i$ can change in total variation distance \cite{D70, DGJ08, DGJ09}. More modern techniques show that it is enough to bound the worst case ``spectral influence" of $j$ to $i$; that is the maximum eigenvalue of the correlation matrix of $X_i, X_j$ in the worst case when all other particles are fixed (for concreteness such a matrix is in $\R^{(S_i\cup S_j)\times (S_i\cup S_j)}$ with entries corresponding to the conditional probabilities, i.e.,  $\P[X_i|X_j=s])$ for entries in the row corresponding to spin $s$ for site $j$). 

Let $\lambda_{i,j}$ be the (spectral) influence of $j$ to $i$ and let $\Lambda$ be the matrix of all pairwise influences; then it has been  shown that if $\|\Lambda\|_{\op}<1$ the natural Markov chain called the Glauber dynamics mixes rapidly, see e.g., \cite{H06}. 

The aforementioned technique has found many applications e.g., in studying graph coloring or multi-state hard core model or even in studying graphical models \cite{D70, DS84, H06, BPCPSDV22, LO25}. Although it is very powerful, it comes short of analyzing many applications of the natural Glauber dynamics. Ideally, one would like an average local-to-global theorem which allows for an ``average spectral influence" of particle $j$ to $i$ instead of taking the worst case spin of all other particles. 
The main result of this manuscript 
is such a technique that we call  
rank-1-perturbed trickledown theorem. 

\paragraph{{\bf Rank-1-perturbed trickledown theorem:}}
Given the $n$-particle spin system, and function $f:S_1\times\dots \times S_n\to\R$ let $P_if$ be the function which averages $f$ on $S_i$ such that the spin of all other particles are fixed, namely:
$$ P_if(x_1,\dots,x_n)=\E_\mu[f(X_1,\dots,X_n) | X_j=x_j,\forall j\neq i].$$

Let $\cP=\sum_i P_i$.
Then, it follows that the transition probability matrix of the Glauber dynamics is $\frac1n{\cal P}.$ We let ${L}_i=I-P_i$ and $\cL =\sum_i{L}_i.$
It is not hard to see that the operator $P_i$ is a projection, i.e., for any function $f$, $P_i^2f=P_if$ and it is self-adjoint with respect to the inner product defined below. 
Similarly, $L_i$'s are projections. Define the inner-product 
$$ \langle f,g\rangle_\mu=\E_{X\sim\mu}[f(X)g(X)].$$
It is not hard to see that $P_i,L_i$ are self-adjoint with respect to this inner-product, i.e., for two functions $f,g$
\[\langle f, P_i g \rangle = \E_{\mu}[f\, \E_{\mu}[g \, | \, X_{-i}]] = \E_{\mu}[ \E_{\mu}[f \, | \, X_{-i}] \,\E_{\mu}[g \, | \, X_{-i}]] = \E_{\mu}[\E_{\mu}[f \, | \, X_{-i}]\, g] = \langle P_if , g \rangle \]
and since $L_i = I-P_i$, the self-adjointness of $L_i$ follows.
Additionally, we have the following equivalent expressions of the Dirichlet form
\begin{equation}\label{eq:dirichlet-form}
    \cE(f) := \E_{\mu}[f\cL f] =  \sum_i\E_{\mu}[f L_if] = \sum_i\E_{\mu}[L_if L_if] = \sum_i \E_{\mu}\big[(L_if)^2 \big].
\end{equation}

\begin{fact}\label{fact:LLvsL}
    If for any function $f:{\cS}\to \R$,
    $$ \langle \cL f,\cL f\rangle\geq \eps\langle \cL f,f\rangle=\eps\cdot \cE(f),$$
    then the spectral gap of the Glauber dynamics $\frac1{n}\cP$ is at least $\eps/n.$
\end{fact}
\begin{proof}
Given an arbitrary function $g$, since $\cL\succeq 0$, we let $f=\cL^{-1/2} g$ (to be precise, we need to take the pseudo-inverse of $\cL$ since $\cL$ has a 0 eigenvalue). 
Having this,
    $$ \langle \cL g,g\rangle=\langle \cL f,\cL f\rangle\geq\eps\cdot\langle \cL f,f\rangle = \langle g,g\rangle,$$
    where all the inner-products are with respect to $\mu$. This concludes the proof.
\end{proof}
We emphasize that the above machinery of using the ``averaged curvature condition'' (\cite[Equation 1.16.6, Proposition 4.8.3]{BGL14}) was recently employed by G\"{o}bel, Jenssen, Michelen, Pappik, Perkins and Schiller \cite{GJMPPS26} to re-prove a recent trickledown theorem of Chen, Chen, Chen, Yin, and Zhang \cite{CCCYZ25}. Their proof was a motivation for us to consider this technique for the SK model.

\begin{definition}\label{def:rank-1-perturbed}
    We say $\cP$ exhibits a rank-1-perturbed trickledown condition if there exists a family of vectors $\{\bm(X)\} \in \R^n,$ for all $X\in {\cS}$ such that for any function $f:{\cS}\to\R$ and any $1\leq i<j\leq n$ and any (feasible) spin $x_{-i-j}$ (for all particles except $i,j$) we have
    \begin{align} \label{eq:rank1perturbed}&\E_{X\sim\mu}[L_if(X)L_jf(X) | x_{-i-j}] \geq -a_{i,j}\E[L_if(X)L_jf(X)m_i(X)m_j(X)|x_{-i-j}] \\
    &\quad- b_{i,j} \E[L_if(X)^2+L_jf(X)^2|x_{-i-j}]
    \end{align}
\end{definition}

The following abstracts our main technique. 
\begin{proposition}\label{prop:rank-1-perturbed}
    Suppose $\cP$ exhibits rank-1-perturbed trickledown condition w.r.t.\ vectors ${\bf m}(\bx)$ with corresponding constants $\{a_{i,j},b_{i,j}\}_{1\leq i,j\leq n}$. If for every $i$,
    \begin{equation}\label{eq:rank-1-perturbed-condition}
        \sum_j b_{i,j} \E[L_if(X)^2]+ \lambda_{\max}(A)\E[m_i(X)^2L_if(X)^2] \leq (1-\eps)\E[L_if(X)^2],
    \end{equation}
    where $A\in \R^{n\times n}$ has entries $a_{i,j}$, then $\langle \cL f, \cL f\rangle\geq \eps\cdot\cE(f)$.
\end{proposition}
\begin{proof}
    Let $w_i(X)=L_if(X) m_i(X)$. Multiplying both sides of \eqref{eq:rank1perturbed} by $\sum_{X\in{\cS}:X_{-i-j}=\sigma_{-i-j}}\mu(X)$ and summing up over all $i,j$ we get
    \begin{align*}
        \langle\cL f,\cL f\rangle&\geq -\E[{\bf w} {\bf w}^\top]\bullet A+ \sum_i (1-\sum_j b_{i,j})\E[\cL_i f(X)^2]
        %n(1-\max_i \sum_j b_{i,j})\langle \cL f,f\rangle
        \\
        &\geq -\lambda_{\max}(A)\tr(\E[{\bf w}{\bf w}^\top])+\sum_i (1-\sum_j b_{i,j})\E[\cL_i f(X)^2]\\
        &= -\lambda_{\max}(A)\cdot \E[\|{\bf w}\|^2]+\sum_i (1-\sum_j b_{i,j})\E[\cL_i f(X)^2]\\
        &=\sum_i \left(\left(1-\sum_j b_{i,j}\right)\E[\cL_i f(X)^2]-\lambda_{\max}(A)\E[m_i(X)^2\cL_if(X)^2]\right)\\
        &\geq \eps\cdot \cE(f),
    \end{align*}
    where the second inequality follows by Lemma~\ref{lem:amgm}. \qedhere
\end{proof}

To put the above statement in context, suppose that we let $a_{i,j}=0$; then, we have to essentially let $b_{i,j}$ be the (spectral) influence of $j$ to $i$, i.e., the worst maximum eigenvalue of the correlation matrix when all other spins are fixed.  Then, the above theorem would imply that if $\sum_j b_{i,j}\leq 1-\eps$, for all $i$, we get a $\eps/n$ spectral gap. The above statement allows to shift part of the ``correlation'' between particles $i,j$ by a common rank-1 matrix $\bm(X)\bm(X)^\top$ as long as  this shift on the average is upper-bounded for every particle $i$. We note that it is also possible to shift the correlation by a sum of several rank-1 matrices as long as we can analogously bound the shifts. We expect to find many future applications of this technique in analyzing multi-state spin systems.

\subsection{Application: The Sherrington-Kirkpatrick (SK) model}
Given a 2-state spin system with $n$ particles, i.e., $S_i=\{-1,+1\}$ for all $i$, we use $J$ to denote the interaction matrix between these particles. 
Let $\mu_J$ be the probability distribution of the corresponding Ising model with this interaction matrix, namely for any vector $\bx\in \{\pm1\}^n$,

\begin{equation}\label{eq:muJ}
     \mu_J(\bx)
  = \frac{1}{Z_n}\exp\bigg(\frac12\bx^\top J \bx\bigg),
\end{equation}
where $Z_n:=\sum_{\bx} \exp(\bx^\top J\bx)$ is the partition function.
Of particular interest to our work is the Sherrington-Kirkpatrick (SK) model in which the interaction matrix $J$ is a random Gaussian ensemble matrix, namely $J$ is a symmetric matrix with zero diagonal and the non-diagonal entries are independent Gaussians with mean zero and variance $\frac{\beta^2}{n}$.
This model is extensively studied recently as a canonical {\em average case} problem in approximate counting and sampling with the hope that the techniques developed here can be extended more generally to sampling problems over random constraint satisfaction problems.

The following is our main theorem.

\begin{theorem}\label{thm:gap-L}
    Let $\cL_J=\sum_i L_i=\sum_i (I-P_i)$. There exists an absolute constant $\eps \geq 5 \cdot 10^{-5}$ with the following property. If $\beta < \tfrac{1}{2} + \eps$, then for sufficiently large $n>n_0$ with high probability over the randomness of the interaction matrix $J$,
    $$\langle \cL_J f,\cL_J f\rangle_\mu\geq C \cdot \langle \cL_J f,f\rangle_\mu$$
    where $C$ is a function of $\eps$.
\end{theorem}

Now, using Fact~\ref{fact:LLvsL} we have the following immediate consequence.

\begin{corollary}\label{cor:gap-L}
     There exists an absolute constant $\eps >0$ with the following property. When $\beta<1/2 + \eps$, the spectral gap of $\frac{1}{n}\cP_J$ is $\Omega(1/n)$, and therefore Glauber dynamics for the SK model mixes in time $O(n^2)$.
\end{corollary}
\begin{proof}
   We will show that with probability $1 - o_n(1)$, we have
   \begin{equation}\label{eq:pi-min}
       \mu_{\min} := \min_{\bx \in \{\pm 1 \}^n}\mu_{J}(\bx)  \geq \exp(-(\ln 2 + 4\beta)n).
   \end{equation}

   With \eqref{eq:pi-min}, it follows from Theorem~\ref{thm:spec-gap} that the mixing time is $O(\log(1/\mu_{\min})/\Omega(1/n)) = O(n^2)$, as desired.
   
  It remains to prove $\eqref{eq:pi-min}$. By \eqref{eq:muJ}, a step towards this is to show that with probability $1 - o_n(1)$, we have for all $\bx \in \{ \pm 1 \}^n$,
  \begin{equation}\label{eq:bounded-energy}
      |\bx^\top J \bx| \leq 2\beta n.
  \end{equation}
 Assuming that we are on the high probability event for which \eqref{eq:bounded-energy} is true, we can bound $Z_n = \sum_{\bx \in \{ \pm 1 \}^n} \exp\Big(\frac12\bx^\top J \bx\Big) \leq 2^n \exp(\beta n) $. Consequently, this gives us \eqref{eq:pi-min}: \[\min_{\bx \in \{\pm 1 \}^n}\mu_{J}(\bx) \geq \frac{\exp(-\beta n)}{2^n \exp(\beta n)}= \exp(-(\ln 2 + 2 \beta)n).\] 
  
  Finally, it remains to prove \eqref{eq:bounded-energy}. We fix $\bx \in \{ \pm 1 \}^n$.  Since $\Var(\bx^\top J \bx)  = \frac{4\beta^2}{n} \sum_{i<j}\Var(g_{ij}x_ix_j) = \frac{4\beta^2}{n}\binom n2 = 2\beta^2(n-1)$ where $g_{ij} \sim \mathcal{N}(0,1)$ for $i,j \in [n]$, we have $\bx^\top J \bx \sim \mathcal{N}(0, 2\beta^2(n-1))$. Consequently, we may apply the Gaussian tail bound to the standard normal random variable $\frac{\bx^\top J x}{\beta\sqrt{2(n-1)}}$ to obtain, %for each $\bx \in \{ \pm 1 \}^n$, that 
  \[\P\bigl[|\bx^\top J\bx|>2\beta n\bigr] \leq \P\Bigg[\frac{|\bx^\top J\bx|}{\beta\sqrt{2(n-1)}}>\sqrt{2n}\Bigg]
\leq 2\exp\left(
        \frac{-(\sqrt{2n})^2}{2}
    \right)  \leq 2e^{-n}.\] 
    Taking a union bound over all $\bx \in \{ \pm 1 \}^n$, 
   \[ \P\big(\max_{\bx \in \{ \pm 1 \}^n}|\bx^\top J \bx| > 2n \big) \leq 2^{n+1} \exp(-n) = o_n(1). \]
   This shows \eqref{eq:bounded-energy}. 
\end{proof}

\paragraph{Concurrent work.} 
This work is done concurrently and independently of the work of Wang \cite{W26} who obtains related results based on the Bochner framework, although our detailed proof techniques are different.

\paragraph{Acknowledgements}
The authors have used AI systems for many heavy computations involved in this project especially in proving Theorem \ref{lem:main} and also for vibe-coding Mathematica to solve for the parameters in the proof of Theorem \ref{thm:gap-L}.
The first author would like to thank Leon Schiller for discussing with him their paper~\cite{GJMPPS26}. This paper is the main inspiration for our work. The first author would also like to thank Thomas Rothvoss, Victor Reis, Arvin Sahami, Dan Mikulincer, Ramon van Handel and Santosh Vempala for discussions on stochastic localization. The second author would like to thank Tselil Schramm and her lab for fruitful discussions on spin systems. The last author's research is supported by an NSF
grant CCF-2203541, a Simons Investigator Award 928589, and a Lazowska Endowed Professorship in Computer Science \& Engineering.

\subsection{Previous Work}

Physicists \cite{SZ81}, \cite{MPV86} predicted that Glauber dynamics for the SK model would converge fast for all $\beta<1$ four decades ago.
For quite some time there were almost no significant computational results in the community. In 2017, Bauerschmidt and Bodineau~\cite{BB19} proved a log-Sobolev inequality for $\mu_{J,h}$ when $J$ is a positive semi-definite matrix with norm $\|J\|_{\op} < 1$ by decomposing the Ising measure as a strongly log-concave mixture of product measures.

A few years later Eldan, Koehler and Zeitouni~\cite{EKZ22} proved that Glauber dynamics for $\mu_{J,h}$ for PSD $J$ with norm at most 1  has spectral gap at least $\frac{C}{n}$, for a universal constant $C>0$, by decomposing the Ising measure into a mixture of Ising measures with rank-one interaction matrices. Consequently, Glauber dynamics mixes in  $O(n^2)$-steps for  SK model when $\beta < \frac14$. 
It is worth noting that the condition $\|J\|_{\op} < 1$ is tight in the worst case for PSD matrices, as the matrix $\beta J$ with all entries in $J$ being $\frac1n$, a.k.a. the Curie--Weiss model, does not mix in polynomial time for $\beta > 1$.

Anari, Jain, Koehler, Pham and Vuong~\cite{AJKPV22} proved a modified log-Sobolev inequality for $\mu_{J,h}$ under the same assumptions using the entropic independence machinery and the needle decomposition result by Eldan, Koehler and Zeitouni \cite{EKZ22}. Consequently they obtained an optimal $O(n \log n)$-mixing time for $\beta < \frac14$ for the SK model. Later, Chen and Eldan~\cite{CE25} reproved this result  using an entropy-annealing theorem and the log-Sobolev inequality result by Bauerschmidt and Bodineau~\cite{BB19}. 

Later Anari, Koehler and Vuong~\cite{AKV24} managed to extend trickledown theorems \cite{Opp18,ALO24} to the stochastic localization machinery, and used it to show that if $J$ is PSD, $\|J\|_{\op} < 4\beta$ and diagonal entries of $J$ are $2\beta$, then Glauber dynamics for $\mu_{J,h}$ satisfies a modified-log Sobolev inequality, for $\beta$ up to approximately $0.295$. 

On a separate line of work,
El Alaoui, Montanari and Sellke~\cite{AMS22} gave an efficient statistical algorithm to sample from the SK model for $\beta$ up to $\frac12$  in Wasserstein distance (as opposed to the  total-variation distance). This was later improved to the full regime of $\beta<1$ by Celentano~\cite{Celetano24}. So, in a sense these results provide a substantial indication that polynomial mixing bound for Glauber dynamics should be reachable.

Very recently, Davies, Lee, Sandhu and Shi~\cite{DLSS26} proved a weak-Poincare inequality for Glauber dynamics for SK model for $\beta < \frac12$. That is, they prove a lower-bound on the spectral gap bound only for ``smooth functions". This implies fast-mixing in total-variation distance from a ``warm start".

\subsection{Proof overview}
By Fact~\ref{fact:LLvsL}, to bound the spectral gap of $\cL$ it suffices to give a lower bound for 
\begin{equation}\label{eq:expandL-overview}
    \frac{\langle \cL f, \cL f \rangle}{\cE(f)} =\frac{\sum_i \E[(L_if)^2] + \sum_{i \neq j} \E[(L_if)(L_jf)]}{\cE(f)} \underset{\eqref{eq:dirichlet-form}}{=} 1 + \frac{\sum_{i \neq j} \E[(L_if)(L_jf)]}{\cE(f)}.
\end{equation}

The term $\E[(L_if)(L_jf)]$ corresponds to the 2-link given by a two-state system corresponding the the spins $x_i$, $x_j$. We further "pin" $x_{-i-j}$, and assume that $\mu_{J_{ij},a,b}(x_i,x_j)=\exp(J_{ij}x_ix_j+ax_i+bx_j)$ is the conditional measure.

\paragraph{$a=b=0$ Case.} First, as a warm-up, we consider the two-state measure on sites $i,j$ as a two state system with interaction $J_{ij}$ on particles $x_i,x_j$: $\mu_{J_{ij}}(x_i,x_j)=\exp(J_{ij}x_ix_j)$.

For any function $f:\{\pm1\}^2$ we prove the following inequality which implies that $\mu$ exhibits a rank-1-perturbed trickledown condition with vectors $m_i(\bx)=x_i$, $a_{i,j}=J_{i,j}$, and $b_{i,j}=|J_{ij}|^3/6$,

\begin{align}\label{eq:rank-1-perturbed}
    \E_{\mu_{J_{ij}}}[(L_if)(L_jf) \, | \, X_{-i-j}] \geq - J_{ij} \E_{\mu_{J_{ij}}}[X_iX_j (L_if)(L_jf) \, | \, X_{-i-j}] - \frac{|J_{ij}|^3}{6}\E_{\mu_{J_{ij}}}\big[(L_if)^2 + (L_jf)^2\,| \, X_{-i-j}\big].
\end{align}

Assuming \eqref{eq:rank-1-perturbed}, we are in a position to apply Proposition~\ref{prop:rank-1-perturbed}. 
 
Using standard facts about GOE random matrices, we have $2 \sum_{j \neq i} b_{ij}  \leq o(1)$ for all $i \in [n]$ (see Lemmas~\ref{lem:row-sum} and ~\ref{lem:max-entry}). This shows that \eqref{eq:rank-1-perturbed-condition} holds for $\eps = 1 - \| J \|_{\op}-o(1)$:
\begin{equation}\label{eq:rank-1-naive-sum}
     \bigg(\frac{1}{6} \sum_j  |J_{ij}|^3 \bigg) \cdot \E[L_if(X)^2]+ \lambda_{\max}(J)\E[X_i^2 L_if(X)^2] \underset{X_i^2 = 1 }{\leq} (1-\lambda_{\max}(J) - o(1))\E[L_if(X)^2].
\end{equation}
Therefore, since $\lambda_{\max}(J) \to  2 \beta + o(1)$ (see Proposition~\ref{prop:Bai-Yin}), it follows that for $\beta < 1/2 - \eps$, $\cL$ has a spectral gap of $\geq \eps$. 

\paragraph{Proof of \eqref{eq:rank-1-perturbed}.}  
For simplicity, we drop the pinned variables $x_{-i-j}$. 

We observe that $\exp(J_{ij}x_ix_j) = \cosh J_{ij}(1 + x_ix_j\sol{J}_{ij})$ where $\sol{J}_{ij} = \tanh J_{ij}$ (see the proof of Lemma~\ref{lem:mu}) which implies $\mu_{J}(x_i,x_j) = \big(1+x_ix_j\sol{J}_{ij} \big)\mu_0(x_i,x_j)$. It follows that
\begin{equation}\label{eq:LiLjfsimplecase} \E_{\mu_{J_{ij}}}[(L_if)(L_jf)]=\E_{\mu_0}[(1+X_iX_j\sol{J})(L_if)(L_jf)].\end{equation}
We don't know a direct lower-bound on this quantity.

Instead, first we use $\mu_{J_{ij}}(X_i = x_i\, | \, X_j = x_j) =  
\frac{1 + \sol{J}_{ij}x_ix_j}{2}$ to write
\begin{align}\label{eq:Li-L0}
    L_if(x_i,x_j) &= f(x_i,x_j) - \E_{X_i}[f(X_i, x_j) \, | \, x_j] \notag\\ &= f(x_i,x_j) -\left( \frac{1 + \sol{J}_{ij}x_ix_j}{2}\cdot f(x_i,x_j) +  \frac{1 - \sol{J}_{ij}x_ix_j}{2}  \cdot f(-x_i, x_j) \right) \notag \\ &=\frac{1 - \sol{J}_{ij}x_ix_j}{2}(f(x_i,x_j) - f(-x_i,x_j)) =  \big(1- \sol{J}_{ij} x_ix_j \big)L^0_{i}f(x_i, x_j)
\end{align}
where $L^0_{i} f(x_i,x_j) = \frac{f(x_i,x_j, x_{-i-j}) - f(-x_i, x_j, x_{-i-j})}{2}$ is the operator corresponding to $\mu_0$ (i.e. the product measure when $J_{ij}=0$ and there is no interaction between the spins at sites $i$ and $j$). Similarly, 
$L_jf(x_i,x_j)= \big(1- \sol{J}_{ij} x_ix_j \big)L^0_{j}f(x_i,x_j)$. 
Having this we can continue from \eqref{eq:LiLjfsimplecase} 
\begin{align*}
    \E_{\mu_{J_{ij}}}[(L_if)(L_jf)]=\E_{\mu_0}[(1+X_i X_j\sol{J}_{ij})(1-X_i X_j\sol{J}_{ij})^2 (L^0_i f)(L^0_j f) ].
\end{align*}
Now, notice if we introduce a $1+X_iX_j\sol{J}$ term in the LHS, the RHS simplifies to $(1-\sol{J}_{ij}^2)^2\E[(L_i^0f)(L_j^0 f)]$. This is the main motivation for the rank-1-perturbation. In particular, we can write
\begin{align}\label{eq:box}
    \E_{\mu_{J_{ij}}}\big[ \big(1+ \sol{J}_{ij}X_iX_j \big) (L_if)(L_jf)\big] 
    \underset{\substack{x_i^2 = x_j^2 = 1 \\ \text{independence}}}{=}\E_{\mu_0}\big[ \big(1 - \sol{J}_{ij}^2\big) (L_{i,0}f)\big]  \E_{\mu_0}\big[ \big(1 - \sol{J}_{ij}^2\big) (L_{j,0}f)\big] \underset{\text{independence}}{\geq} 0.
\end{align}

Having figured out the right rank-1 shift, we have 
\begin{align*}
    &\E_{\mu_{J_{ij}}}[(L_if)(L_jf) \, | \, \bX_{-i-j}]  \\&\underset{\text{AM-GM}}{\geq}  -J_{ij} \E_{\mu_{J_{ij}}}[X_iX_j (L_if)(L_jf) \, | \, \bX_{-i-j}] -\frac{|J_{ij} - \sol{J}_{ij}|}{2} \E_{\mu_{J_{ij}}}\big[(L_if)^2 + (L_jf)^2 \, | \, \bX_{-i-j}\big]
    \\
    &\geq - J_{ij} \E_{\mu_{J_{ij}}}[X_iX_j (L_if)(L_jf) \, | \, \bX_{-i-j}] - \frac{|J_{ij}|^3}{6}\E_{\mu_{J_{ij}}}\big[(L_if)^2 + (L_jf)^2\,| \, \bX_{-i-j}\big],
\end{align*}
showing \eqref{eq:rank-1-perturbed}. 

\paragraph{General Case: $a,b$ are not necessarily zero.}
Now, we consider the more general case for 
$\mu_{J_{ij},a,b}$ with local fields $a= \sum_{k \neq i,j}J_{ik}x_k$ (respective $b= \sum_{k\neq i,j}J_{jk} x_k$) experienced by site $i$ (respective site $j$) with respect to the pinned states $\bx_{-i-j}$. In this case we need to lower bound
$\E_{\mu_{J_{ij},a,b}}[L_if \cdot L_jf]$. We prove the following inequality in Lemma~\ref{lem:main}.
   \begin{align}\label{eq:local-ineqb} &\E\big[(L_i f)(L_jf)\big] + J_{ij}\cdot \E\big[m_i(x_i,x_j)m_j(x_i,x_j)(L_if)(L_jf)\big] + \frac{(\lambda J_{ij})^2}{2}\cdot \E\big[(L_if)^2
    + (L_jf)^2\big]\notag\\&+ \frac{dJ_{ij}^2}{2}\E\big[\sol{\bX}_i(x_j)^2\cdot L_if(x_i,x_j)^2 + \sol{\bX}_j(x_i)^2\cdot  L_jf(x_i,x_j)^2\big] \geq 0.
    \end{align}
    where $\sol{{\bf X}}_i$ is the conditional expectation of $X_i$ given $\bx_{-i}$, $\lambda=o(1)$ is a small constant (that upper-bounds $J_{ij}$ entry) and $d,s$ are universal constants that one can optimize. 
    Note that the corresponding ${\bf m}(\bx)$ vectors are defined by $m_i(\bx)=x_i+s\sol{\bx}_i.$ 
    
    In this case we write
    $$\mu_{J_{i,j},a,b}(x_i,x_j) (1+ \sol{J}_{ij} \sol{X}_{i,0} \sol{X}_{j,0}) \underset{\text{Lemma~\ref{lem:Radon-Nikodym}}}{=} \mu_{0,a,b}(x_i,x_j) \big(1 + \sol{J}_{ij} x_ix_j \big),$$ where $\sol{X}_{i,0} = \E_{\mu_{0,a,b}}[X_i] = \tanh a$ (see \eqref{eq:tanh}; particularly, $\sol{X}_{i,0} = 0$ when $a=0$). Using this, in analogy with \eqref{eq:Li-L0} we have \[\big(1+ \sol{J}_{ij}\sol{X}_{i,0} x_j \big) L_i f(x_i, x_j) = \big(1 - \sol{J}_{ij} x_i x_j \big) L^0_{i}f(x_i,x_j).\] 
    
    Now, we run a calculation similar to \eqref{eq:LiLjfsimplecase}, which gives
    \begin{align*}
        \E_{\mu_{J_{ij},a,b}}[(L_if)(L_jf)] = \E_{\mu_{0,a,b}}\Bigg[(1-\sol{J}_{ij}^2)\frac{1-\sol{J}_{ij}X_iX_j}{(1+\sol{J}_{ij}\sol{X}_{i,0}\sol{X}_{j,0})(1+\sol{J}_{ij}\sol{X}_{i,0}X_j)(1+\sol{J}_{ij}\sol{X}_{j,0}X_i)}(L^0_if)(L^0_jf)\Bigg].
    \end{align*}
    This suggests that we should multiply the LHS with $(1+\sol{J}_{i,j}X_iX_j)(1+\sol{J}_{ij}\sol{X}_{i,0}\sol{X}_{j,0})(1+\sol{J}_{ij}\sol{X}_{i,0}X_j)(1+\sol{J}_{ij}\sol{X}_{j,0}X_i)$. We write
    \begin{align*}
     &\E_{\mu_{J_{ij},a,b}}[(1+(X_i+\sol{X}_i)(X_j+\sol{X}_j)+O(\sol{J}_{ij}^2))(L_if)(L_jf)] =\\
     &\underset{|\sol{J}_{ij}|\leq1/2}{\geq}\E_{\mu_{J_{ij},a,b}}[(1+\sol{J}_{i,j}X_iX_j)(1+\sol{J}_{ij}\sol{X}_{i,0}\sol{X}_{j,0})(1+\sol{J}_{ij}\sol{X}_{i,0}X_j)(1+\sol{J}_{ij}\sol{X}_{j,0}X_i)(L_if)(L_jf)]\\
     &=(1-\sol{J}_{ij}^2)^2\E_{\mu_0}[L^0_if]\E_{\mu_0}[L^0_jf]\geq 0,
    \end{align*}
    where the last inequality similarly follows from $\mu_0$ being an independent distribution.

    Using AM-GM the above rearranges to the following analogous inequality as \eqref{eq:rank-1-perturbed}
    \begin{align}\label{eq:rank-1-perturbeda}
    &\E_{\mu_{{J_{ij}},a,b}}[(L_if)(L_jf) \, | \, \bX_{-i-j}] \\
    &\geq - J_{ij} \E_{\mu_{J_{ij}}}[(X_i+\sol{\bX}_i(X_j)) L_if\cdot  (X_j+\sol{\bX}_j(X_i)) L_jf \, | \, \bX_{-i-j}] - O\big(\sol{J}_{ij}^2 \big)\E_{\mu_{J_{ij}}}\big[(L_if)^2 + (L_jf)^2\,| \, \bX_{-i-j}\big].
\end{align}

In comparison to the earlier computations in \eqref{eq:rank-1-perturbed}, the error terms here are of order $O\big( \sol{J}_{ij}^2 \big)$ rather than $O\big( \sol{J}_{ij}^3 \big)$. When we take a sum over all the links, by Lemma~\ref{lem:row-sum}, we have $\sum_{j\ne i}J_{ij}^2=\beta^2+o(1)$ for all $i \in [n]$, while as we saw earlier Lemma~\ref{lem:max-entry} implies that $\sum_{j\ne i}|J_{ij}|^3=o(1)$ for all $i \in [n]$. That is, while we could dispense with the $O\big(\big|\sol{J}_{ij}\big|^3\big)$ error terms in \eqref{eq:rank-1-naive-sum}, the quadratic error contributes a non-negligible constant order term. 

For this reason, we work with a slightly different rank one term in Lemma~\ref{lem:main}, where we take $m_i(\bx)=x_i+s\sol{\bx}_i$ for some $s <1$ instead of $x_i + \sol{\bx}_i$ as suggested by our computations above; the $(1-s)\sol{\bx}_i$ difference here will be used to compensate for the aforementioned $O\big( \sol{J}_{ij}^2 \big)$ error term. As a result, the proof of Lemma~\ref{lem:main} (see Section~\ref{sec:local}) is a fair bit more involved than the preceding sketch. 

\paragraph{$\beta <1/2$ Case.}  Having \eqref{eq:local-ineqb}, and averaging over all pairs $i,j$ in a way similar to the proof of Proposition~\ref{prop:rank-1-perturbed}, we essentially obtain (again omitting some error terms)
\[ \langle \cL f, \cL f \rangle \geq (1 - \|J\|_{\op} -o_n(1)) \cE(f) + \underbrace{s(2-s) \|J \|_{\op} \sum_i |\sol{\bX}_i|^2(L_if)^2}_{(A)} + \underbrace{\E[\bw^\top (\|J \|_{\op} I - J)\bw]}_{(B)}. \]

At this point, we note that $A, B \geq 0$. If we were to just drop these terms, then because $\|J \|_{\sf op} \to 2 \beta + o_n(1)$, we would obtain a spectral gap $\eps/n$ for Glauber dynamics in the range $\beta < \tfrac{1}{2}-\eps$. The range $\beta <\tfrac{1}{2}-\eps$ is in fact tight in the context of proving a spectral gap for Glauber dynamics on Ising models for complete graphs with interaction matrix $J$ such that $\|J\|_{\op} = 2 \beta$. This can be seen from the Curie-Weiss model with the convention $J_{ij} = \beta/n$ for $i \neq j$ and $J_{ii} = 0$ for all $i \in [n]$. It is known that the Curie-Weiss model undergoes a phase transition at $\beta = 1/2$, after which Glauber dynamics no longer has a spectral gap of order $\Omega(1/n)$; more precisely, \cite{DLP09,LLP10}\footnote{We note that these works have a different normalization for the inverse temperature parameter $\beta$. These works also establish a stronger notion of ``cut-off phenomenon'' for the Curie-Weiss mdoel.} show that when $\beta > 1/2$ the spectral gap is $\exp(-\Omega(n))$.

\paragraph{$\mathbf{\beta<1/2+\eps}$ Case. }
To prove a spectral gap for Glauber dynamics for the SK model in the range of $\beta \in [1/2-\eps, 1/2 + \eps)$, we need to use additional properties of the GOE interaction matrix (beyond simply the spectral radius fact $\|J\|_{\op} \to 2\beta +o_n(1)$) to show $A \gtrsim \cE(f)$ or $B \gtrsim  \cE(f)$.

To that end, we use two properties of GOE random matrices. First, that their eigenvectors are delocalized, namely that the mass of the eigenvector is spread across all coordinates $i\in [n]$ rather than being concentrated on few (see Lemma~\ref{lem:sparse-delocalization}). Second, since the entries of GOE random matrices are i.i.d. Gaussians, we can also show that w.h.p. for all $\bx \in \{ \pm 1 \}^n$, $(1- \delta)$-fraction (think of $\delta$ small) of sites $i \in[n]$ are such that $\sol{\bX}_i = \tanh\left(\sum_{k \neq i} J_{ik}x_k\right)$ is large in magnitude (see Lemma~\ref{lem:typ-large}); we denote these sites $H_{>}$. Notice that if we could show that $\sol{\bX}_i$ is large for all $i$, then we would be immediately done. 

As such, we consider the following two cases:
\begin{enumerate}
\item [(ii)] If $f$ is chosen such that $L_if$ is small for $i \in H_{>}$, then since $w_i \approx L_if$, this would imply that most of the mass of $\bw$ is concentrated on the $\delta n$ coordinates $i \not \in H_{>}$. By the contrapositive of Lemma~\ref{lem:sparse-delocalization}, this would imply that $\bw$ does not correlate well with any eigenvector of $J$, and in particular the one corresponding to the largest eigenvalue. Therefore, in this case we get \[B \gtrsim \|\bw\|_2^2  \gtrsim \sum_{i=1}^n(L_if)^2 \underset{\eqref{eq:dirichlet-form}}{\geq}2\eps\cdot  \cE(f).\]
    \item [(i)] Suppose $B \approx 0$; that is, $\bw$ is close to the direction of the largest eigenvalues of $J$. By Lemma~\ref{lem:sparse-delocalization}, $\bw$ is delocalized. Since $w_i = (X_i + \sol{\bX}_i)\cdot  L_i f$, this would also imply that $L_if$ are all roughly equal in magnitude. Consequently, \[ A = s(2-s)\|J \|_{\op}\sum_i |\sol{\bX}_i|^2(L_if)^2 \gtrsim \left( \frac{1}{n}\sum_{i \in H_{>}}|\sol{\bX}_i|^2 \right) \sum_i(L_if)^2 \underset{\eqref{eq:dirichlet-form}}{\geq} 2\eps \cdot \cE(f). \]
    
\end{enumerate}
We refer the reader to Section~\ref{sec:proof-main} where we carry out this dichotomy to prove a spectral gap for $\cL$ up to $\beta \leq \tfrac{1}{2} + 5 \cdot 10^{-5}$. We did not optimize this part of the proof, and expect that it is possible to extend our techniques to obtain a spectral gap for $\cL$ in an even larger range of $\beta$.

\section{Preliminaries}\label{sec:prelim}
\subsection{Notation and basic facts}

We will use capital letters to denote random variables and bold letters to denote vectors.
For a vector $\bx\in \R^n$ we write
$$ \norm{\bx}^2_{\leq m}:=\min_{S\subseteq [n], |S|=m} \sum_{i\in S} x_i^2$$
 We say that the matrix $J \in \R^{n \times n}$ is \emph{$\beta$-GOE} if $J$ is symmetric, and $J_{ij} \stackrel{\sf i.i.d.}{\sim} \mathcal{N}(0, \beta^2/n)$ for $i \neq j$ and $J_{ii} = 0$ for $i \in [n]$. We use $\Phi(x)$ to denote the CDF of the standard normal random variable. 

\begin{lemma}\label{lem:amgm}
Let $C \succeq 0$ and $A \in \R^{n \times n}$ be a symmetric matrix. Then 
\[ A \bullet C \leq  \lambda_{\max}(A) \Tr(C).\]
\end{lemma}
\begin{proof}
    Since $\lambda_{\max}(A)I - A \succeq 0$, it follows that $C^{1/2}(\lambda_{\max}(A)I - A)C^{1/2} \succeq 0$ as well. Consequently, 
    \[ 0 \leq \Tr\big(C^{1/2}(\lambda_{\max}(A)I - A)C^{1/2}\big) = \Tr((\lambda_{\max}(A)I - A)C) =  \lambda_{\max}(A) \Tr(C) - A \bullet C. \qedhere \]
\end{proof}

\begin{fact}\label{fact:L-one-dim}
    Let $X\in \{\pm1\}$ be a random variable with expectation $\sol{X}$. For a function $f:\{-1,1\}\to\R$, let $L$ be defined as $Lf(x)=f(x)-\E_X f(X)$. Then, there is a constant $\kappa$ (a function of $f$) such that
    $$ Lf(X)=\kappa\cdot(X-\sol{X}).$$
\end{fact}

\begin{proof}
    In fact, we claim that we may take 
    \[ \kappa = \frac{f(1) - f(-1)}{2}.\]
   We also define $c = \kappa + f(-1)$, so that $f(x) = c+ \kappa x$ for $x \in \{ \pm 1 \}$. This in turn implies that 
   \[ Lf(X) = f(X) - \E[f(X)] = \big(c+\kappa X \big) - \big(c + \kappa \sol{X} \big) = \kappa \big(X - \sol{X} \big).\qedhere\]
\end{proof}

\subsection{Ising measure}\label{sec:Ising}

In this section, we collect some basic facts about the two-state Ising measure: let $J, a,b \in \R$ and we define the measure to be $\mu_{J,a,b}(x,y) = \frac{1}{Z} \exp(Jxy + a x + b y)$ for $x,y\in \{ \pm 1 \}$.

At $J=0$, we can check that we have the product measure 
\[ \mu_{0,a,b}(x,y) = \frac{\exp(ax)}{2 \cosh a} \cdot \frac{\exp(bx)}{2 \cosh b}.\]
Because of this factorization, it follows that 
\begin{equation}\label{eq:tanh}
     \E_{\mu_{0,a,b}}[X] = \tanh a =: \sol{X}_0,\quad \E_{\mu_{0,a,b}}[Y] = \tanh b =: \sol{Y}_0 \quad\text{and} \quad \E_{\mu_{0,a,b}} [XY] = \sol{X}_0 \sol{Y}_0
\end{equation}

For simplicity, we define $\sol{J} = \tanh J$ and also let 
\[ \sigma_{X_0} = \sqrt{1 - \sol{X}_0^2} \quad\text{and}\quad \sigma_{Y_0}= \sqrt{1 - \sol{Y}_0^2}.\]
We also define 
\[\sol{\bX}(y) = \E[X  \, | \, y]\quad\text{and}\quad \sol{\bY}(x) = \E[Y \, | \, x].\]

  \begin{lemma}\label{lem:mu}
         For $\mu=\mu_{J,a,b}$, we have \[\mu(x,y) = \frac{(1+\sol{X}_0x)(1+ \sol{Y}_0y)(1+\sol{J} xy)}{4(1+\sol{J}\sol{X}_0\sol{Y}_0 )}, \quad \overline{\bX}(y) = \frac{\sol{X}_0 + \sol{J}y}{1+ \sol{X}_0 \sol{J}y} \quad\text{and}\quad \sol{\bY}(x) = \frac{\sol{Y}_0 + \sol{J} x}{1 + \sol{Y}_0 \sol{J}x}\]
     \end{lemma}  
 \begin{proof}
          First, for $x\in \{\pm1\}$, 
          \begin{equation}\label{eq:exp-cosh-snh}
              \exp(cx) = \cosh (c)+ x \sinh(c) = \cosh(c)(1 + x \tanh(c)).
          \end{equation}
          Recall that $\sol{X}_0=\tanh(a),\sol{Y}_0=\tanh(b)$, so
    \[\exp(ax) = \cosh(a) \big(1+\sol{X}_0 x\big), \quad \exp(by) = \cosh(b) \big(1+ \sol{Y}_0 y\big) \quad\text{and}\quad \exp(Jxy) = \cosh(J) \big(1 + \sol{J} xy).\]
    The numerator of $\mu(x,y)$ is given by $ \cosh(a)\cosh(b)\cosh(J)(1+ \sol{X}_0 x)(1+ \sol{Y}_0 y)(1+ \sol{J} xy)$. Now, since 
    \[ \sum_{x,y \in \{ \pm 1 \}}(1+ \sol{X}_0 x)(1+ \sol{Y}_0 y)(1+ \sol{J} xy) = 4(1+\sol{J}\sol{X}_0\sol{Y}_0 ), \]
    it follows that 
    \[ \mu(x,y) = \frac{(1+\sol{X}_0x)(1+ \sol{Y}_0y)(1+\sol{J} xy)}{4(1+\sol{J} \sol{X}_0\sol{Y}_0)}.\]
    Next, we also have 
    \[ \sum_{x \in \{ \pm 1 \}} (1+ \sol{X}_0x)(1 + \sol{J}xy) = 2 (1+ \sol{X}_0\sol{J} y),\]
    so that 
    \[ \overline{\bX}(y) = \sum_{x \in \{ \pm 1 \}} \frac{(1+ \sol{X}_0x)(1+\sol{J}xy)}{2(1+\sol{X}_0\sol{J} y)} \cdot x = \frac{\sol{X}_0 + \sol{J}y}{1+ \sol{X}_0 \sol{J}y}. \]
    Symmetrically we get the result for $\sol{\bY}(x)$.
     \end{proof}

\subsection{Taylor expansion of two-state Ising model around zero interaction}

We do a Taylor expansion of the two-state spin system for the Ising measure and various functions of the underlying spins, when the interaction $J$ is very close to 0. This helps us later on when we want to prove a rank-1 perturbed trickle-down condition (Definition~\ref{def:rank-1-perturbed}) for a conditioned two-state spin system.

\begin{lemma}\label{lem:bfx}
    For all $x \in \{\pm 1\}$ and all $y \in \{\pm1 \}$ and $\big|\sol{J} \big| \leq \frac{1}{2}$, we have $\sol{\bX}(y) = \sol{X}_0 + \sol{J}\sigma_{X_0}^2y  \pm  O(|J|^2)$ and $\sol{\bY}(x) = \sol{Y}_0 + \sol{J} \sigma_{Y_0}^2 x \pm O(|J|^2)$.
\end{lemma}

 \begin{proof}
         By the expression that we have computed in Claim~\ref{lem:mu}, it follows that 
         \begin{align*}
             \big|\sol{\bX}(y) - \sol{X}_0 - \sol{J} \sigma_{X_0}^2y\big| &= \bigg| \frac{\sol{X}_0 + \sol{J} y}{1 + \sol{J} \sol{X}_0 y} - \sol{X}_0 - \sol{J} (1- \sol{X}_0^2) y \bigg|\\
             &= \bigg|\frac{\sol{J} y(1-\sol{X}_0^2)(1 - (1+\sol{J}\sol{X}_0 y))}{1 + \sol{J} \sol{X}_0 y} \bigg|\\
             &= \bigg|\frac{\sol{J}^2 y^2(1-\sol{X}_0^2) \sol{X}_0}{1 + \sol{J} \sol{X}_0 y} \bigg|\underset{|y|,|\sol{X}_0|\leq 1,\sol{J}\leq 1/2}{\leq} 2 \sol{J}^2
         \end{align*}
         
         for small enough $|J|$, so that $\sol{\bX}(y) = \sol{X}_0 + \sol{J} \sigma_{X_0}^2 y \pm O(|J|^2)$ as desired. The result for $\sol{\bY}(x)$ follows by symmetry. 
    \end{proof}

\begin{lemma}\label{lem:Radon-Nikodym}
    For $|\sol{J}|\leq 1/2$ we have 
    \[ \frac{\mu_{J,a,b}(x,y)}{\mu_{0,a,b}(x,y)} = \frac{1+\sol{J} xy}{1 + \sol{J} \sol{X}_0 \sol{Y}_0}=1 + \sol{J} \big(xy - \sol{X}_0\sol{Y}_0 \big) + \mu^R_{J,a,b}(x,y)\]
\end{lemma}
where $|\mu^R_{J,a,b}(x,y)| = O\big(\sol{J}^2 \big)$ for all $x,y \in \{ \pm 1 \}$. 
\begin{proof}
    We first prove the first equality. By an analysis similar to \eqref{eq:exp-cosh-snh}, we have 
    \[ \exp(Jxy) = (\cosh J)\cdot (1+ \sol{J} xy).\]
    In particular, this means that we may write 
    \begin{align*}
        Z_{J,a,b} &= \cosh J\sum_{x,y \in \{ \pm 1 \}} \exp(ax+ by)\big(1 + \sol{J} xy \big) \\
        &= (\cosh J) Z_{0,a,b} \E_{\mu_{0,a,b}}(1+\sol{J}xy) \\
        &= (\cosh J)Z_{0,a,b}(1+\sol{J}\sol{X}_0 \sol{Y}_0).
    \end{align*}
    Consequently, it follows that
\[ \mu_{J,a,b}(x,y) = \frac{\exp(Jxy + ax + by)}{Z_{J,a,b}} = \frac{(\cosh J)(1+ \sol{J} xy) \exp(ax +by)}{Z_{0,a,b} (\cosh J)(1+\sol{J} \sol{X}_0 \sol{Y}_0)} = \frac{\mu_{0,a,b}(x,y) (1+\sol{J}xy)}{1 + \sol{J} \sol{X}_0 \sol{Y}_0}. \]
    Rearranging this proves the first equality of the lemma.

    Next, we prove the second equality. To that end, we note that 
    \begin{equation}\label{eq:rad-nik-first}
        \Bigg|\frac{\mu_{J,a,b}(x,y)}{\mu_{0,a,b}(x,y)} -1 \Bigg|=  \Bigg|\frac{1+ \sol{J}xy}{1 + \sol{J} \sol{X}_0 \sol{Y}_0} - 1 \Bigg|= \Bigg|\frac{\sol{J}\big(xy - \sol{X}_0 \sol{Y}_0 \big)}{1+ \sol{J} \sol{X}_0 \sol{Y}_0}\Bigg|
    \end{equation}
    and so 
    \[\big | \mu_{J,a,b}^R(x,y) \big| := \Bigg|\frac{\mu_{J,a,b}(x,y)}{\mu_{0,a,b}(x,y)} -1 - \sol{J}\big(xy - \sol{X}_0 \sol{Y}_0 \big) \Bigg| = \sol{J}^2 \Bigg|\big(xy - \sol{X}_0\sol{Y}_0 \big) \frac{\sol{X}_0 \sol{Y}_0}{1 + \sol{J} \sol{X}_0 \sol{Y}_0} \Bigg| \underset{|\sol{J}|\leq 1/2}{\leq} 4 \sol{J}^2\]
     which rearranges to the second equality.
\end{proof}
.
\begin{lemma}[Master Lemma]\label{lem:master}
    Let $F \colon \R \times \{ \pm 1 \}^2 \to \R$ be a twice-differentiable function such that
    \begin{equation}\label{eq:F-Taylor}
        F(\sol{J} ,x,y) = F(0,x,y) + \sol{J} \dot{F}(0,x,y) +  \sol{J}^2 \ddot{F}(\sol{J}^*(x,y),x,y)
    \end{equation}
    for some $\sol{J}^*(x,y)$ between $0$ and $\sol{J}$, where $\dot{F}(0,x,y)=\frac{d}{d\sol J}F\big(\sol{J},x,y\big)\bigg|_{\sol{J} = 0}$ and $\ddot{F}(\sol{J}^*(x,y),x,y)=\frac{d^2}{d\sol{J}^2}F\big(\sol{J},x,y \big)\biggr|_{\sol{J} = \sol{J}^*(x,y)}.$ 
    If for some constant $C>0$ and all $x,y\in \{\pm1\}$ and $|\sol{J}|\leq 1/2$ we have
    \begin{equation}\label{eq:F-conditions}
        \E_{\mu_{0,a,b}}\big[|F(0,x,y)|\big], \;
        \E_{\mu_{0,a,b}}\big[\big|\dot{F}(0,x,y)\big|\big],\;\E_{\mu_{0,a,b}}\big[\big|\ddot{F}(\sol{J}^*(x,y),x,y) \big| \big] \leq C  
    \end{equation}
    then for $\big|\sol{J} \big|\leq \tfrac{1}{2}$, we have 
    \[ \E_{\mu_{J,a,b}}\big[F\big(\sol{J},x,y \big)\big] = \E_{\mu_{0,a,b}}[F(0,x,y)] + \sol{J} \big( \E_{\mu_{0,a,b}}[\dot{F}(0,x,y)] + \E_{\mu_{0,a,b}}[(XY - \sol{X}_0 \sol{Y}_0) F(0,x,y) \big] \big) \pm O(C \cdot \sol{J}^2).\]
\end{lemma}

\begin{proof}
We begin by remarking that the form of $F$ in \eqref{eq:F-Taylor} follows from Taylor's theorem. For simplicity of the notation we write $F(J)$ to denote $F(J,x,y)$.
    We now write 
    \begin{align*}
        \E_{\mu_{J,a,b}}[F(J)] &= \E_{\mu_{0,a,b}} \bigg[\frac{\mu_{J,a,b}(x,y)}{\mu_{0,a,b}(x,y)} \big( F(0) + \sol{J} \dot{F}(0) + \sol{J}^2 \ddot{F}\big(\sol{J}^* \big) \big) \bigg] \\
        &\underset{\text{Lemma~\ref{lem:Radon-Nikodym}}}{=}\E_{\mu_{0,a,b}} \big[ \big(1 + \sol{J} \big(xy - \sol{X}_0\sol{Y}_0 \big) + \mu_{J,a,b}^R(x,y) \big) \big( F(0) + \sol{J} \dot{F}(0) + \sol{J}^2 \ddot{F}\big(\sol{J}^* \big) \big) \big) \big].
    \end{align*}
    The constant term of the above is $\E_{\mu_{0,a,b}}[F(0)]$, the linear term in $\sol{J}$  is 
    \[ \sol{J} \big(\E_{\mu_{0,a,b}}\big[ \dot{F}(0) \big] + \E_{\mu_{0,a,b}} \big[\big(XY - \sol{X}_0 \sol{Y}_0\big) F(0)\big ] \big). \]

    The quadratic term in $\sol{J}$ is 
    \begin{align*}
    &\E_{\mu_{0,a,b}}\big[ \mu^R_{J,a,b} \cdot F(0) \big] + \sol{J}^2\E_{\mu_{0,a,b}}\left[ (xy - \sol{X}_0 \sol{Y}_0)  \dot{F}(0) \right] + \E_{\mu_{0,a,b}}\big[ \big(1 + \sol{J} \big(xy - \sol{X}_0\sol{Y}_0 \big) + O\big(\sol{J}^2) \big) \sol{J}^2 \ddot{F}\big(\sol{J}^* \big) \big)\big] \\
        & \hspace{-0.5em}\underset{\substack{|x|,|y|,\big|\sol{X}_0 \big|, \big|\sol{Y}_0\big| \leq 1\\ \text{Lemma~\ref{lem:Radon-Nikodym}}: \, \big|\mu^R_{J,a,b} \big| \leq O\big( \sol{J}^2\big)}}{\leq}O\big(\sol{J}^2\big)\E_{\mu_{0,a,b}}[ |F(0)|] +  2\sol{J}^2\E_{\mu_{0,a,b}}\left[   \big|\dot{F}(0)\big| \right] + \big(1 + 2\sol{J} + O\big(\sol{J}^2) \big) \sol{J}^2 \E_{\mu_{0,a,b}}\big[\big|  \ddot{F}\big(\sol{J}^*  \big) \big|  \big]\\
        &\underset{\eqref{eq:F-conditions}, \, \big|\sol{J}\big| \leq 1/2}{\leq} O\big(C \cdot \sol{J}^2\big). \qedhere
    \end{align*}
\end{proof}

\subsection{GOE random matrices}

\begin{proposition}[{\cite{BY88}}]\label{prop:Bai-Yin}
Let $\beta>0$, and let $J$ be $\beta$-GOE. Then $\|J \|_{\op}\to 2 \beta$ in probability.
\end{proposition}
\begin{lemma}\label{lem:row-sum}
    Let $J\in \R^{n\times n}$ be the $\beta$-GOE matrix. Then, for every $\sum_j J_{ij}^2\to  \beta^2.$
\end{lemma}
\begin{proof}
    Observe that 
    $J_{i1},\dots,J_{in}$ are independent Gaussians with variance $\beta^2/n$. Therefore $\sum_j J_{ij}^2$ is highly concentrated around its expected value which is $\beta^2.$
\end{proof}

\begin{lemma}\label{lem:max-entry}
Let $\beta >0$ ane $J$ be $\beta$-GOE. Then \[\max_{1\le i,j\le n}|J_{ij}| \to 0\] in probability.
\end{lemma}

\begin{proof}
Fix $t>0$. Then the Gaussian tail bound gives
\[
  \mathbb P\bigl(|J_{ij}|>t\bigr)
  =
  \mathbb P_{Z \sim \mathcal{N}(0,1)}\left(|Z|>\frac{t\sqrt n}{\beta}\right)
  \le
  2\exp\left(-\frac{nt^2}{2\beta^2}\right).
\]

By taking a union bound over the $\frac{n(n-1)}{2}$ off-diagonal entries, we get 
\begin{align*}
  \mathbb P\left(
    \max_{1\le i,j\le n}|J_{ij}|>t
  \right) =
  \mathbb P\left(
    \max_{1\le i<j\le n}|J_{ij}|>t
  \right) \le
  \sum_{1\le i<j\le n}
  \mathbb P\bigl(|J_{ij}|>t\bigr)\le
  \frac{n(n-1)}{2}\exp\left(-\frac{nt^2}{2\beta^2}\right),
\end{align*}
which converges to zero as $n\to\infty$.
\end{proof}

\begin{lemma}[GOE eigenvector delocalization]\label{lem:sparse-delocalization}
    Let $\delta \in (0,1)$ and $\rho, \beta > 0$. Suppose $J$ is $\beta$-GOE. Then with high probability over the randomness in $J$, for all $\bz \in \R^n$, we have 
        \[ \bz^\top (2 \beta I - J)\bz \geq \frac{\pi \beta \rho}{6.1(\beta + \rho)}\delta^3 \cdot \| \bz \|_2^2 - \rho \norm{\bz}^2_{\leq \delta n}.\]
\end{lemma}

One way of interpreting the above statement is that if $\bz$ is an eigenvector of $J$ with eigenvalue close to $\lambda_{\max}(J) \approx 2 \beta$, then we expect $\bz$ to have its mass ``spread out'' among its coordinates and in particular have substantial mass in every subset of $\delta n$ coordinates. 

\begin{proof}[Proof of Lemma~\ref{lem:sparse-delocalization}]
    For $\bz \in \mathbb{S}^{n-1}$, we define the centered Gaussian processes 
    \[ X_{\bz} = \bz^\top J \bz \quad\text{and}\quad Y_{\bz} = \frac{2\beta}{\sqrt{n}} \bg^\top \bz \quad\text{where} \; \bg\sim \mathcal{N}(0,I_n). \]

    We note that for $\bz,\tilde{\bz} \in \mathbb{S}^{n-1} $ that 
        \begin{align*}
        \E\big[(X_{\bz}-X_{\tilde{\bz}})^2 \big] &=4\,\E\bigg[\big(\sum_{i<j}J_{ij}(z_iz_j-\tilde{z}_i\tilde{z}_j)\big)^2\bigg]\\&\underset{}{=} \frac{4 \beta^2}{n} \sum_{i < j} (z_i z_j - \tilde{z}_i\tilde{z}_j)^2 \leq \frac{2\beta^2}{n} \big\| \bz\bz^\top - \tilde{\bz}\tilde{\bz}^\top \big \|_F^2 \\
        &= \frac{4\beta^2}{n} \big[\big(1 - (\bz^\top \tilde{\bz})^2 \big) \big] \underset{1-x^2 \leq 2(1-x)}{\leq} \frac{4\beta^2}{n} \big[ 2\big(1-\bz^\top \tilde{\bz} \big) \big] \\
        &= \frac{4 \beta^2}{n} \|\bz - \tilde{\bz} \|_2^2 = \E \big[ (Y_{\bz} - Y_{\tilde{\bz}})^2 \big].
    \end{align*}

    \begin{theorem}[Variant of Sudakov-Fernique comparison inequality {\cite{V00}}]\label{thm:Sudakov-Fernique}
    Let $\cI$ be a countable index set. Suppose $\{X_i\}_{i \in \cI}$ and $\{Y_i \}_{i \in \cI}$ are two mean-zero Gaussian processes  such that 
    \[ \E(X_i - X_j)^2 \leq \E(Y_i - Y_j)^2 \; \forall i,j\in\cI.\]
    Let $\{m_i \}_{i \in \cI}$ be such that $m_i \in \R$ for $i \in \cI$ independent of the underlying probability measure. Then 
    \[ \E \sup_i (X_i + m_i) \leq \E \sup_i (Y_i + m_i)\; \forall i\in\cI.\]
\end{theorem}

    Applying above to $\{X_{\bz}\}_{\bz \in \mathbb{S}^{n-1}}$, $\{Y_{\bz} \}_{\bz \in \mathbb{S}^{n-1}}$ and $m_{\bz} = -\rho \norm{\bz}^2_{\leq \delta n}$ for $\bz \in \mathbb{S}^{n-1}$, we obtain 
    \begin{equation}\label{eq:quad-to-lin}
        \E \sup_{\bz \in \mathbb{S}^{n-1}} \bigg \{ X_{\bz} - \rho \norm{\bz}^2_{\leq \delta n} \bigg \} \leq \E \sup_{\bz \in \mathbb{S}^{n-1}} \bigg \{ \frac{2\beta}{\sqrt{n}}g^\top \bz - \rho \norm{\bz}^2_{\leq \delta n } \bigg\}.
    \end{equation}
    We note that $\mathbb{S}^{n-1}$ is not enumerable but we can apply the theorem to a countable dense subset of $\mathbb{S}^{n-1}.$ First, we bound the RHS:
\begin{claim}\label{claim:expectedXz}
    $$ \E \sup_{\bz \in \mathbb{S}^{n-1}} \bigg \{ \frac{2\beta}{\sqrt{n}}\bg^\top \bz - \rho  \norm{\bz}^2_{\leq \delta n} \bigg\} \leq 2\beta-\frac{\pi \beta\rho}{120(\beta+\rho)}\delta^3$$
\end{claim}
\begin{proof}
Let $\bar{\bg}=\bg/\sqrt{n}$. First, using $\norm{\bz}=1$ we observe that
\begin{align}\label{eq:gtzbound}
    2\beta\bar{\bg}^\top \bz - \rho \norm{\bz}^2_{\leq \delta n}=\beta(1+\norm{\bar{\bg}}^2)-(\beta\norm{\bar{\bg}-\bz}^2+\rho\norm{\bz}^2_{\leq \delta n})\leq \beta(1+\norm{\bar{\bg}}^2)-\frac{\beta\rho}{\beta+\rho}\norm{\bar{\bg}}^2_{\leq \delta n}
\end{align}
To see the inequality, let $S$ be the smallest $\lfloor\delta n \rfloor$ coordinates of $\bz$ (in absolute value). Then, by triangle inequality,
$$ \norm{\bar{\bg}}_{\leq \delta n} = \sqrt{\sum_{i\in S} \bar{g}_i^2} \underset{\text{triangle ineq}}{\leq} \sqrt{\sum_{i\in S} (\bar{g}_i-z_i)^2} + \sqrt{\sum_{i\in S} z_i^2} \leq \norm{\bar{\bg}-z} + \norm{\bz}_{\leq \delta n}  $$
Therefore, by the Cauchy-Schwarz inequality, 
$$ \norm{\bar{\bg}}^2_{\leq \delta n} \leq \bigg(\frac{1}{\beta}+\frac{1}{\rho} \bigg)(\beta\norm{\bar{\bg}-\bz}^2+\rho\norm{\bz}_{\leq \delta n}^2)$$
Having \eqref{eq:gtzbound} and using $\E\norm{\bar{\bg}}^2=1$, we have that
$$\E \sup_{\bz \in \mathbb{S}^{n-1}} \bigg \{ \frac{2\beta}{\sqrt{n}}\bg^\top \bz - \rho  \norm{\bz}^2_{\leq \delta n} \bigg\} \leq 2\beta-\frac{\beta\rho}{\beta+\rho}\frac{\E\norm{\bg}^2_{\leq \delta n}}{n} 
$$
So, to finish the proof it is enough to show that $\E\norm{\bg}^2_{\leq \delta n}\geq n \frac{\pi\delta^3}{6}.$ 

For every $t > 0$, 

\begin{align*}
    \frac{\E \|\bg \|_{\leq \delta n}^2}{n} &\geq \frac{\E \min_{S \subset [n]:|S|=\lfloor\delta n \rfloor} \sum_{i\in S} \min \{ g_i^2 ,t^2 \}}{n} \\
    &\geq \frac{\E \min_{S:|S| =\lfloor\delta n \rfloor} \bigg\{\delta n t^2 - \sum_{i=1}^n (t^2 - g_i^2)_+ \bigg \}}{n} \\
    &= \delta t^2 - \E_{g \sim \mathcal{N}(0,1)} (t^2 - g^2)_+ = \delta t^2 - \int_{-t}^t (t^2 - x^2) \frac{\exp(-x^2/2)}{\sqrt{2\pi}}dx \\
    &\underset{e^{-x^2/2}\leq 1}{\geq} \delta t^2 - \frac{4t^3}{3\sqrt{2} \pi} \underset{t = \sqrt{\pi/2}\delta}{= }\frac{\pi}{6} \cdot \delta^3. 
\end{align*}
where we write $(a)_+=\max\{a,0\}.$
\end{proof}

Next, we use the following theorem to prove that the LHS is concentrated around its expectation.

    \begin{theorem}[Borell-TIS inequality {\cite{Bor76,TIS76}}, see {\cite{Sel24}} for non-centered version]\label{thm:Borell}
        Let $\{X_i\}_{i\in {\cal I}}$ be a not necessarily centered Gaussian process. Suppose for every $i\in {\cal I}$, $\Var X_i \leq \sigma^2.$ Then, 
        $$ \P \big(\max_{i\in {\cal I}} X_i - \E\big[\max_{i\in {\cal I}} X_i]\geq t \big)\leq \exp(-t^2/(8\sigma^2)).$$
    \end{theorem}

For convenience let $X_{\bz}' = X_{\bz} - \rho \|\bz \|_{\leq \delta n}^2 $. Then, 
\begin{equation}\label{eq:VarXz}
    \Var(X_{\bz}')=\E X_{\bz}^2=\E\Big[\Big(\sum_{i<j}2J_{ij}z_iz_j\Big)^2\Big]=\frac{2\beta}{n}\sum_{i,j} 2z_i^2z_j^2\leq \frac{2\beta}{n}\norm{\bz}^4=\frac{2\beta}{n}.
\end{equation}

Consequently, this implies that 
    \begin{align}\label{eq:high-prob-Xz}
        &\PP\bigg(\sup_{\bz \in \mathbb{S}^{n-1}} \big\{  X_{\bz}' \big\} \geq 2\beta -\frac{\pi \beta \rho}{6.1(\beta + \rho)} \delta^3\bigg) \notag \\&\underset{\text{Claim~\ref{claim:expectedXz}}}{\leq}  \PP\bigg(  \sup_{\bz \in \mathbb{S}^{n-1}} \big\{  X_{\bz}' \big\} - \E \sup_{\bz \in \mathbb{S}^{n-1}} \big\{  X_{\bz}'\big\}\geq \frac{\pi \beta \rho}{36.6(\beta + \rho)} \delta^3\bigg) \notag \\&\hspace{-2em}\underset{\text{Theorem~\ref{thm:Borell} and}\eqref{eq:VarXz}}{\leq} \exp\left( -\left( \frac{\pi \beta \rho}{14520(\beta + \rho)} \delta^3 \right)^2 \frac{n}{36.6\beta}  \right).
    \end{align}
Consequently, it follows that with probability $1 - o_n(1)$, we have 
\begin{align*}
    \bz^\top (2 \beta I - J) \bz =  2 \beta - X_{\tilde{\bz}}
    \underset{\eqref{eq:high-prob-Xz}}{\geq} \frac{\pi \beta \rho}{6.1(\beta + \rho)}\delta^3 - \rho \|\bz \|_{\leq \delta n}^2,
\end{align*}
with probability $1 - o_n(1)$. This proves the lemma. 
\end{proof}

In the context of Sherrington-Kirkpatrick $\mu_{J, \bh}$ with external field $\bh \in \R^n$, we define $H_i(\bx) = h_i + \sum_{j \neq i}J_{ij} x_j$ to be the local field felt by spin $i$. 

\begin{lemma}[Typical large local fields]\label{lem:typ-large}
Let $\beta>0$ and $J\in \R^{n\times n}$ be $\beta$-GOE. Suppose $0<\delta' < \delta <1 $ is such that 
\begin{equation}\label{eq:delta-delta-prime}
    D_{\mathrm{KL}}(\mathrm{Ber}(\delta) \parallel \mathrm{Ber}(\delta')) > \ln 2.
\end{equation}

Let $\bh \in \R^n$ be arbitrary. With high probability, for all $\bx \in \{\pm1 \}^n$, we have 
\[ \bigg| \bigg\{ i: |H_i(\bx)\big| \leq \beta \Phi^{-1}\big(\tfrac{1+\delta'}{2} \big)\bigg\} \bigg| < \delta n.\]
\end{lemma}

\begin{proof}
  Fix a vector $\bx\in \{\pm1\}^n$. We will prove that
  \begin{equation}\label{eq:large-local-field}
      \P[|\{i:|H_i(\bx)|\leq\beta\Phi^{-1}(\tfrac{1+\delta'}{2})\}| \geq \delta n] \le \exp\Big(-nD_{\mathrm{KL}}\Big(\mathrm{Ber}(\delta) \parallel \mathrm{Ber}\Big( \delta'+o_n(1)\Big)\Big) \Big)
  \end{equation}

  Now running a union bound over all $2^n$ many $\bx$ vectors, we obtain that
  $$ \P[\exists\bx: |\{i:|H_i(\bx)|\leq\beta\Phi^{-1}(\tfrac{1+\delta'}{2})\}| \geq \delta n] \le 2^n \exp\Big(-nD_{\mathrm{KL}}\Big(\mathrm{Ber}(\delta) \parallel \mathrm{Ber}\Big( \delta'+o_n(1)\Big)\Big) \Big)\underset{\eqref{eq:delta-delta-prime}}{\leq} o(1),$$
where the last inequality uses the assumption that $D_\mathrm{KL}(\mathrm{Ber}(\delta) \parallel \mathrm{Ber}( \delta'))>\ln 2$.
  
 It remains to prove \eqref{eq:large-local-field}. 
The covariance matrix of
$(H_1(\bx),\ldots,H_n(\bx))$ is
\begin{equation}\label{eq:local-field-covariance}
\Cov\bigl(H_1(\bx),\ldots,H_n(\bx)\bigr)
=\frac{\beta^2}{n}\left((n-2)I_n+\bx\bx^{\top}\right).
\end{equation}

Note that $\E[H_i(\bx)^2]=\beta^2(n-1)/n=\beta^2(n-2)/n+x_i^2.$
Consequently, if $Z_0,Z_1,\ldots,Z_n$ are independent $\mathcal{N}(0,1)$ random variables, then $(H_1(\bx),\ldots,H_n(\bx))$ has the same law as
\begin{equation}\label{eq:local-field-representation}H_i'(\bx) = \bigg(h_i + \beta\sqrt{\frac{n-2}{n}}\,Z_i
+\frac{\beta}{\sqrt n}x_iZ_0 \bigg)_{i\in[n]}.
\end{equation}
So, to prove \eqref{eq:large-local-field} it is enough to show that 
\begin{equation}\label{eq:large-local-fielda}\P[\big| \big\{ i: |H_i'| \leq \beta \Phi^{-1}\big(\tfrac{1+\delta'}{2} \big)\big\} \big| \geq \delta n] \le \exp\Big(-nD_{\mathrm{KL}}\Big(\mathrm{Ber}(\delta) \parallel \mathrm{Ber}\Big( \delta'+o_n(1)\Big)\Big) \Big).
\end{equation}

For $Z_i,Z_0 \sim \mathcal{N}(0,1)$, since the distribution of $\beta\sqrt{(n-2)/n}+\beta/\sqrt{n}x_i Z_0$ is symmetric around the origin, the function
$$h_i\mapsto \P\bigg[\bigg| h_i+\beta\sqrt{\frac{n-2}{n}}Z_i+\frac{\beta}{\sqrt{n}} x_iZ_0 \bigg|\leq \beta\Phi^{-1}\big(\tfrac{1+\delta'}{2}\big)\bigg]$$ 
is  nonincreasing for $h_i\ge0$.  Since conditioned on $Z_0$, the random variables $H'_i(\bx)$ are independent, for a fixed $Z_0$, the probability of the above event is maximized when for all $i$, $h_i=0$. By a similar argument, the probability is maximized when $Z_0=0$. So, let $H''_i=\beta\sqrt{\frac{n-2}{n}}Z_i$ and note that it is independent of $\bx$. To prove \eqref{eq:large-local-fielda} it is enough to show that 
\begin{equation}\label{eq:H''bound} \P[\big| \big\{ i: |H_i''| \leq \beta \Phi^{-1}\big(\tfrac{1+\delta'}{2} \big)\big\} \big| \geq \delta n]\leq \exp\Big(-nD_{\mathrm{KL}}\Big(\mathrm{Ber}(\delta) \parallel \mathrm{Ber}\Big( \delta'+o_n(1)\Big)\Big) \Big).\end{equation}
Let $Y_i$ be the indicator random variable that $H''_i>\beta\Phi^{-1}\big(\tfrac{1+\delta'}{2} \big)$. Then for each $i \in [n]$, we have the Gaussian tail bound
\[ \PP(Y_i) \leq 2\Phi\left(\Phi^{-1}\big(\tfrac{1+\delta'}{2} \big)\sqrt{\frac{n}{n-2}}\right)-1 = \delta'+ o_n(1).\]
It follows by the Chernoff bound that 
\[ \P\bigg[ \sum_{i=1}^nY_i \geq \delta n \bigg] \underset{\delta' < \delta}{\leq} \exp\Big(-nD_{\mathrm{KL}}\Big(\mathrm{Ber}(\delta) \parallel \mathrm{Ber}\Big( \delta'+o_n(1)\Big)\Big) \Big).  \]
This proves \eqref{eq:H''bound} which concludes the proof.
\end{proof}

\section{Proof of Theorem~\ref{thm:gap-L}}\label{sec:proof-main}
In this section we prove Theorem \ref{thm:gap-L}. All expectations in this section are with respect to $\mu_J$ unless otherwise specified. We also fix $f \colon \{ \pm 1 \}^n \to \R$. 

First, we observe that
\begin{equation}\label{eq:LfLf-proof}
\langle \cL_J f,\cL_J f\rangle=\sum_{i} \E[(L_i f)^2] + \sum_{i \neq j} \E[L_ifL_j f] = \sum_i \E[(L_i f)^2] + \sum_{i \neq j} \E_X \big[ \E_{X_{-i-j}}[L_ifL_jf]], 
\end{equation}

where $\E_{X_{-i-j}}$ is with respect to $\mu_{X_{-i-j}}$ is the conditional measure on $(X_i,X_j)$ after pinning $X_{-i-j}$. In particular, for $\bx\in \{\pm1\}^n$ we have
 \begin{align}\label{eq:mu-two-site}
     \mu_{\bx_{-i,-j}}(x_i, x_j) &= \frac{\mu(\bx)}{\sum_{X_i,X_j \in \{ \pm 1 \}}\mu(\bx_{-i-j}, X_i, X_j)} \propto \exp\Big( \frac12 J_{ij}x_ix_j + \frac12 x_i\sum_{k \neq i,j}J_{ik}x_k  + \frac12 x_j\sum_{k \neq i,j}J_{jk}x_k  \Big), 
 \end{align}

The following theorem is our local rank-1-perturbation lemma; here, we rename $X_i,X_j$ to $X,Y$ to simplify notation ($L_X,L_Y$ are defined accordingly). We note that it has an additional $\E\big[\sol{\bX}(y)^2\cdot L_Xf(x,y)^2 + \sol{\bY}(x)^2\cdot  L_Yf(x,y)^2\big]$ term compared to the rank-1-perturbation trickledown condition (Definition~\ref{def:rank-1-perturbed}).
We remark that the following lemma is the most challenging part of our proof.
\begin{theorem}[Main Technical Lemma]\label{lem:main}
    Let $d,s\leq 1$ and $\lambda > 0$ be such that 
    \begin{equation}\label{eq:d-s-constraint}
         d \geq \max \{ -(s^2 + 2s -1), 2s^2 \}
    \end{equation}
    Then there exists $j(d,\lambda,s)$ with the following property. For all $J$ such that $|J| \leq j(d,\lambda,s)$ and $a,b \in \R$, 

    For all $f \colon \{ \pm 1\}^2\to \R$, we have 
    \begin{align}\label{eq:local-ineq} &\E\big[(L_X f)(L_Yf)\big] + J\cdot \E\big[M(L_Xf)(L_Yf)\big] + \frac{(\lambda J)^2}{2}\cdot \E\big[(L_Xf)^2
    + (L_Yf)^2\big]\notag\\&+ \frac{dJ^2}{2}\E\big[\sol{\bX}(y)^2\cdot L_Xf(x,y)^2 + \sol{\bY}(x)^2\cdot  L_Yf(x,y)^2\big] \geq 0.
    \end{align}
    where 
        $M(x,y) = \big(x+ s \sol{\bX}\big) \left(y+ s\sol{\bY} \right)$.
\end{theorem}
We defer the proof of Theorem~\ref{lem:main} to Section~\ref{sec:local}. But we would like to point out that the main advantage of the above statement is that it has no dependency on the external fields $a,b$ on the two state Ising measure $\mu_{J,a,b}.$

We will later choose $|d|, |s| \leq 1$ satisfying \eqref{eq:d-s-constraint}. By Observation~\ref{lem:max-entry}, all entries of $J$ are at most $j(d, \lambda, s)$ so by applying Theorem~\ref{lem:main} to $\mu_{x_{-i-j}}$,  we have
\begin{align}\label{eq:x-i-j} &\E_{x_{-i-j}}\big[(L_i f)(L_jf)\big] + J_{ij}\cdot \E_{x_{-i-j}}\big[M_{ij}(L_if)(L_jf)\big] + \frac{(\lambda J_{ij})^2}{2}\cdot \E_{x_{-i-j}}\big[(L_if)^2
    + (L_jf)^2\big] \notag \\&+ \frac{d J_{ij}^2}{2}\E_{x_{-i-j}}\big[\sol{\bX}_i^2\cdot (L_if)^2 + \sol{\bX}_j^2\cdot  (L_jf)^2\big] \geq 0.
\end{align}
For $1\leq i\leq n$ we write $\norm{J_i}^2=\sum_j J_{i,j}^2$.
Averaging the above inequality with respect to $\mu_{-i,-j}(X_{-i,-j})$ and summing over all $i,j \in [n]$, we obtain 
\begin{align}\label{eq:main}
    \sum_{i \neq j} \E[(L_if)(L_jf)]&\geq -\sum_{i \neq j} J_{ij}\E[M_{ij}(L_if)(L_jf)] - \lambda^2\sum_i\norm{J_i}^2\E[L_if^2] - 2d\sum_i \norm{J_i}^2\E[\sol{\bX}_i^2 L_if^2]\notag\\&\hspace{-1.7em}\underset{\text{Lemma}~\ref{lem:row-sum}}{\geq}-\sum_{i \neq j} J_{ij}\E[M_{ij}(L_if)(L_jf)] - \lambda^2(\beta^2+o(1)){\cal E}(f) - d(\beta^2+o(1))\sum_i \E[\sol{\bX}_i^2 L_if^2]\notag\\
    &=:-\sum_{i \neq j} J_{ij}\E[M_{ij}(L_if)(L_jf)] - \lambda^2(\beta^2+o(1)){\cal E}(f) - d(\beta^2+o(1)){\cal D}(f)
\end{align}
We note that we lose a factor of 2 for the two terms in the RHS of the first inequality  because we sum \eqref{eq:x-i-j} over all un-ordered pairs $i,j$.

\begin{lemma}\label{lem:Term1}
Let $\tilde{J} = \| J \|_{\op} I - J$ and $\mathbf{w} \colon \{ \pm 1 \}^n \to \R^n$ is defined as follows: for each $1\leq i\leq n$, we let $w_i \colon \{ \pm 1 \}^n \to \R$ be defined by
\begin{equation}\label{def:w}w_i(x_1,\dots,x_n) = (x_i + s\sol{\bX}_i(\bx_{-i}))\cdot L_i f(\bx).
\end{equation}
Then,
    \begin{equation}\label{eq:Term1}
    \sum_{i \neq j} J_{ij}\E[M_{ij}(L_if)(L_jf)] = \| J \|_{\op} \, \big(\cE(f) - s(2-s) \cD(f) \big) + \E\big[\mathbf{w}^\top \tilde{J} \mathbf{w}\big]
    \end{equation}
\end{lemma}
\begin{proof}
Given $\bx\in\{\pm1\}^n$, let $W(\bx)=\bw\bw^\top$. The main observation is that 
$$ \sum_{i\neq j} M_{ij}(\bx) J_{ij}L_if(\bx)L_jf(\bx)=J\bullet W(\bx),$$
where we are crucially using that the matrix $M$ with entries $M_{ij}=(x_i+s\sol{\bX}_i(\bx_{-i}))(x_j+s\sol{\bX}_j(x_{-j}))$ is a rank 1 matrix.
Letting  $W=\E_{\bx\sim\mu_J} W(\bx)$ by linearity of expectation we have
\begin{align*} \sum_{i \neq j} J_{ij}\E[M_{ij}(L_if)(L_jf)] &= J\bullet \E[W(\bx)] = (\|J\|_{\op}I - \tilde{J}) \bullet W = \|J\|_{\op}(I \bullet W) - \tilde{J}\bullet W \\&= \|J\|_{\op}\E_{\mu_{\beta}}\big[\|w \|_2^2 \big] - \E[\bw^\top \tilde{J} \bw].
\end{align*}
Therefore, to finish the proof it is enough to show that
   $ \E_{\mu_{\beta}} \big[\|w \|_2^2\big] = \cE(f) - s(2-s) \cD(f)$
This can be done by applying the following Claim~\ref{claim:w-computation}
to $g\equiv 1$ and summing up 
over all $i \in [n]$.
\end{proof}

\begin{claim}\label{claim:w-computation}
    Let $\mathbf{w}$ be defined as in Lemma~\ref{lem:Term1} and  $g \colon [-1,1] \to \R$. For $i \in [n]$, we have
    \[ \E\big[g\big(\sol{\bX}_i \big) w_i(X)^2 \big] =  \E\big[ g\big(\sol{\bX}_i \big)\big(1-s(2-s) \sol{\bX}_i^2\big)(L_if)^2\big].\]
\end{claim}
\begin{proof}
    By Fact~\ref{fact:L-one-dim}, it follows that there exists some function $\kappa(\bx_{-i})$ (which is also a function of $f$) for which
\begin{equation}\label{eq:one-dim-L}
L_if(x) = \kappa(\bx_{-i})(x_i - \sol{\bX}_i(\bx_{-i})).
\end{equation}
Consequently, it follows that 
\[ w_i(\bx) = (x_i + s \sol{\bX}_i(\bx_{-i}))L_if(\bx) = \kappa(\bx_{-i})\big(1 -s \sol{\bX}_i(\bx_{-i})^2+(s-1)\sol{\bX}_i(\bx_{-i})x_i\big).\]
That is, we have 
\begin{align*}
    \E\big[g(\sol{\bX}_i)w_i^2 \, | \, \bX_{-i} = \bx_{-i}\big]  &= g(\sol{\bX}_i)\kappa(\bx_{-i})^2 \E\bigg[ \big(1 - s \sol{\bX}_i(\bx_{-i})^2 + (s-1)\sol{\bX}_i(\bx_{-i})X_i \big)^2 \, \bigg| \bX_{-i} = \bx_{-i} \bigg] \\
    &= g(\sol{\bX}_i)\kappa(\bx_{-i})^2  \big(1 - \sol{\bX}_i(\bx_{-i})^2 \big)\big(1 - s(2-s)\sol{\bX}_i(\bx_{-i})^2 \big) \\
    &=\E[g(\sol{\bX}_i)(L_if)^2|\bX_{-i}=\bx_{-i}]\big(1 - s(2-s)\sol{\bX}_i(\bx_{-i})^2 \big)\\
    &= \E\big[g\big(\sol{\bX}_i \big)\big(1-s(2-s) \sol{\bX}_i^2)(L_if)^2 \big | \, \bX_{-i} = \bx_{-i}\big].
\end{align*}

Taking another expectation and using the tower property gives the desired conclusion.
\end{proof}

Putting equations \eqref{eq:LfLf-proof} and \eqref{eq:main} and Lemma \ref{lem:Term1} together we obtain
\begin{align*}\langle \cL_Jf,\cL_Jf \rangle &\underset{\eqref{eq:dirichlet-form}}{=}\cE(f) + \sum_{i \neq j} \E\big[ (L_if) (L_jf)\big] \\&\hspace{-1.7em}\geq \big(1 - \|J \|_{\op} - \lambda^2 \beta^2 - o(1) \big) \, \cE(f) + \underbrace{\E\big[\mathbf w^\top \tilde{J} \mathbf w\big] + \big( s(2-s)\|J \|_{\op} -  d\beta^2 -o(1)\big)\,\cD(f)}_{\cG(f)}, 
\end{align*}
where $\mathbf w \colon \{ \pm 1 \}^n \to \R^n$ is defined as in Lemma~\ref{eq:Term1}. Let $\eps'>0$ be a constant that we choose later.

{\bf Case 1: $\beta\in(0,\tfrac{1}{2} -\eps)$.}
Since $\tilde{J}\succeq 0$ and $\cD(f) \geq 0$, it follows by Proposition~\ref{prop:Bai-Yin} that if $d,s$ satisfy $s(2-s)(2\beta) - d \beta^2 > 0$ and \eqref{eq:d-s-constraint}, then we have for large enough $n$ that 
\[\langle\cL_J f, \cL_J f \rangle \geq \big(1 - 2 \beta - \lambda^2 \beta^2 - o(1) \big) \, \cE(f) \underset{\beta\leq (1-\eps)/2, \,\lambda\text{ sufficiently small}} > \eps \cdot {\cE}(f)\]

It can be checked for instance that $ s = \tfrac{1}{3}$ and $d = \tfrac{2}{9}$ satisfies $s(2-s)(2\beta) - d \beta^2 > 0$ and \eqref{eq:d-s-constraint}. 
This concludes the proof of Theorem~\ref{thm:gap-L} in this case.

{\bf Case 2: $\tfrac{1}{2} - \eps \leq \beta \leq \tfrac{1}{2}+ \eps$.} We claim that for every $\beta \leq \tfrac{1}{2} + \eps$, with high probability we have
 \begin{equation}\label{eq:bound-GE}
     \cG(f) >  (2 \eps + 9 \cdot 10^{-6}) \cdot \cE(f). 
 \end{equation}

 We first show how to complete the proof assuming \eqref{eq:bound-GE}. By Proposition~\ref{prop:Bai-Yin}, with high probability we have $\| J \|_{\op} = 2 \beta \pm o(1)$, which gives 
  \[\E\big[ (\cL_J f)^2\big] \geq \big(1 - 2 \beta - \lambda^2\beta^2 + 2\eps + 9 \cdot 10^{-6} \big) \cdot \cE(f) \geq \big(9 \cdot 10^{-6}-\lambda^2 \beta^2 \big) \cdot \cE(f),\]
    which would complete the proof of Theorem~\ref{thm:gap-L} in this case by taking $\lambda$ sufficiently small.

  It remains to justify \eqref{eq:bound-GE}. For the remainder of the proof, we work in the high probability event that Proposition~\ref{prop:Bai-Yin}, Lemma~\ref{lem:sparse-delocalization} and Lemma~\ref{lem:typ-large} hold. We will choose $0<\delta'<\delta$ later such that $D_{\mathrm{KL}}(\mathrm{Ber}(\delta) \parallel \mathrm{Ber}(\delta')) > \ln 2$  so that we may apply Lemma~\ref{lem:typ-large}.  
  
  Fix $\mathbf{x} \in \{ \pm 1 \}^n$ arbitrary. We define 
  \[ H_{\leq}(\bx) = \big\{ i: |H_i(\bx)| \leq \beta \Phi^{-1} \big( \tfrac{1+\delta'}{2}\big) \big\},\]
  and we define $H_{>}(\bx) = [n]\smallsetminus H_{\leq}(\bx)$.
  
  By Lemma~\ref{lem:typ-large},  $|H_{\leq}(\bx)| \leq \delta n$, so we have, 
  \begin{equation}\label{eq:min-w}
      \| \mathbf{w}(\bx) \|_{\leq (1-\delta) n}^2 \leq \sum_{i \in H_{>}(\bx)} w_i(\bx)^2,
  \end{equation}

Next, by Lemma~\ref{lem:sparse-delocalization} with $\beta$ and $\rho > 0$ to be chosen later applied to the first term in $\cG(f)$, we get 
\[ \mathbf{w}^\top(\bx) \tilde{J} \mathbf{w}(\bx) \geq (c - o(1)) \| \mathbf{w} (\bx) \|_2^2  - \rho   \norm{\bw(\bx)}^2_{\leq(1-\delta)n} \underset{\eqref{eq:min-w}}{\geq} (c - o(1))   \| \mathbf{w} (\bx) \|_2^2  - \rho \sum_{i \in H_{>}(\bx)} w_i(\bx)^2,\]
where $c = c(\beta, \rho, \delta) = \frac{\pi \beta \rho}{6.1(\beta + \rho)}(1-\delta)^3 $ is as in Lemma~\ref{lem:sparse-delocalization}.

Now, taking expectations in the above, we obtain
  \begin{equation}\label{eq:quad-Jbar}
      \E\big[ \mathbf{w}^\top \tilde{J} \mathbf{w} \big] \geq (c - o(1))\E\big[\| \mathbf{w} \|_2^2 \big] - \rho \E\bigg[ \sum_{i \in H_{>}(\bx)} w_i^2\bigg].
  \end{equation}

  We write,
  \begin{align}\label{eq:G-bound} \cG(f) &= \E\big[\mathbf w^\top \tilde{J} \mathbf w\big] + \big( s(2-s)\|J \|_{\op} -  d\beta^2 - o(1)\big)\,\cD(f) \notag \\
      &\underset{\eqref{eq:quad-Jbar}}{\geq} (c-o(1))\E[ \|\mathbf{w} \|_2^2] - \rho \E \bigg[ \sum_{i} \big(\one\{ i \in H_{>}(\bx)\} \cdot w_i^2 \big)\bigg] + (2s(2-s)\beta - \beta^2 d -o(1)) \sum_i \E\big[\sol{\bX}_i^2 (L_if)^2 \big]  \notag  \\ 
    &\underset{\text{Claim}~\ref{claim:w-computation}}{\geq} (c-o(1))\E[ \|\mathbf{w} \|_2^2] - \rho \sum_i \E \bigg[ \one\big\{ H_i(\bx) > \beta \Phi^{-1}(\tfrac{1+\delta'}{2}) \big\} (1- s(2-s) \sol{\bX}_i^2) (L_if)^2 \bigg] \notag \\[-0.25em]&\hspace{7.5em} + (2s(2-s)\beta - \beta^2 d - o(1)) \sum_i \E\big[\sol{\bX}_i^2 (L_if)^2 \big] \notag 
    &\intertext{\text{using $\sol{\bX}_i(\bx) = \tanh H_i(\bx)$,  $\sol{\delta} =  \tanh\big(\beta \Phi^{-1}\big(\tfrac{1+\delta'}{2}\big)\big)$, $\tanh(x)$  strictly monotone for $x>0$ we write}} \notag 
    & \geq (c-o(1))\E[ \|\mathbf{w} \|_2^2] - \rho \sum_i \E \bigg[ \one\big\{ \sol{\bX}_i > \sol{\delta}) \big\} (1- s(2-s)\sol{\bX}_i^2) (L_if)^2 \bigg] \notag \\[-0.25em]
&\hspace{7.5em} + (2s(2-s)\beta - \beta^2 d - o(1)) \sum_i \E\big[\sol{\bX}_i^2 (L_if)^2 \big]  \notag \\ 
&\underset{\text{Claim}~\ref{claim:w-computation}}{\geq} \sum_i (c-o(1))\E\big[(1-s(2-s)\sol{\bX}_i^2 \big) (L_if)^2 \big] - \rho \sum_i \E \bigg[ \one\big\{ \sol{\bX}_i > \sol{\delta} \big\} (1- s(2-s) \sol{\bX}_i^2) (L_if)^2 \bigg] \notag \\[-0.25em]
&\hspace{7.5em} + (2s(2-s)\beta - \beta^2 d - o(1)) \sum_i \E\big[\sol{\bX}_i^2 (L_if)^2 \big]  \\
&> (2 \eps + 9 \cdot 10^{-6}) \cdot \cE(f).
  \end{align}
  
To see the last inequality we set $ \eps=5\cdot10^{-5},\,
    s=\tfrac23,\, d=\tfrac89,\,
    \delta=0.642,\, \delta'=0.13319$, and $
    \rho=0.004693$.

    First, we observe that $\max\{ - (s^2 + 2s -1), 2s^2 \} = \max \{ -\tfrac{7}{9}, \tfrac{8}{9}\} = d$ and $D_{\mathrm{KL}}(\mathrm{Ber}(0.642)\parallel \mathrm{Ber}(0.13319)) > 0.69317 > \ln  2$. Next, we can also check that 
\begin{equation}\label{eq:parameter-1}
2s(2-s)\beta-d\beta^2 =\frac89(2\beta-\beta^2)\underset{\text{monotonic in } \beta}{\geq}
    \frac89\left(
        2\left(\frac12-\eps\right)
        -\left(\frac12-\eps\right)^2
    \right)
    =0.66662222.
\end{equation}
\begin{equation}\label{eq:parameter-2}
0.004693=\rho >c(\beta,\rho,\delta) \underset{\text{monotonic in }\beta}{>} \frac{\pi\cdot0.49995\cdot0.004693}
         {6.1(0.49995+0.004693)}
    (0.358)^3>0.0001098
    >2\eps+9\cdot10^{-6}.
\end{equation}
\begin{equation}\label{eq:parameter-3}
\sol{\delta}^{2} =
    \tanh^2\left(
        \beta\Phi^{-1}\left(\frac{1+\delta'}2\right)
    \right) \underset{\tanh \text{ is monotonic}}{\geq} 
    \tanh^2\left(
        0.49995\cdot\Phi^{-1}(0.566595)
    \right)
    >0.00699758.
\end{equation}

For $i \in [n]$ such that $\big|\sol{\bX}_i\big| \leq \sol{\delta}$, the coefficient of $(L_if)^2$ in \eqref{eq:G-bound} is 
\begin{align}\label{eq:small}
c+
    \left(
        2s(2-s)\beta-\beta^2d-cs(2-s)
    \right)\big|\sol{\bX}_i\big|^2&\underset{\eqref{eq:parameter-1}}{\geq}
    c+
    \underbrace{\left(
        0.66662222-\frac89\cdot0.004693
    \right)}_{>0}
    \big|\sol{\bX}_i\big|^2 \notag \\&\ge c
    \underset{\eqref{eq:parameter-2}}{\geq}0.00010986
    >2\eps+9\cdot10^{-6}.
\end{align}

For $i \in [n]$ such that $\big|\sol{\bX}_i\big| > \sol{\delta}$, the coefficient of $(L_if)^2$ in \eqref{eq:G-bound} is 
\begin{align}\label{eq:large}
&(c-\rho)
    +\left(
        2s(2-s)\beta-\beta^2d-s(2-s)(c-\rho)
    \right)\big|\sol{\bX}_i\big|^2 = (c-\rho)\left(1-\frac89 \sol{\bX}_i^2 \right)     +\left(2s(2-s)\beta-\beta^2d\right)\sol{\bX}_i^2 \notag\\ &\hspace{2em}\underset{\eqref{eq:parameter-1}, \, \eqref{eq:parameter-2}}{\geq}
    -0.00458314+0.67069612 \cdot \big|\sol{\bX}_i\big|^2 \underset{\eqref{eq:parameter-3}}{\geq}
    0.000110109 >
    2\eps+9\cdot10^{-6}.
\end{align}

Summing up \eqref{eq:small} and \eqref{eq:large} over all $i \in [n]$ and using $\sum_i \E[(L_if)^2] = \cE(f)$ (see \eqref{eq:dirichlet-form}), it follows that $\cG(f) > (2\eps + 9 \cdot 10^{-6}) \cdot \cE(f)$, establishing \eqref{eq:bound-GE}.

\section{Proof of the local inequality}\label{sec:local}

In this section, we prove Theorem~\ref{lem:main}. We begin with 
decomposing a given function $f:\{\pm1\}^2\to\R$.
Recall $\mu_{0,a,b}$ corresponds to the Ising measure with no interaction between the spins $x$ and $y$. 

For the random variable $X$ (and similarly for $Y$), we define
$$ \phi(X) = \frac{X-\sol{X}_0}{\sigma_{X_0}},$$
where $\sol{X}_0$ is the expectation of $X$ and $\sigma_{X_0}$ is the standard deviation of $X$ when there is no interaction between $X$ and $Y$.

Note that by definition $\E_{\mu_{0,a,b}}{\phi(X)}=\E_{\mu_{0,a,b}}{\phi(Y)}=0$. Furthermore, $\E_{\mu_{0,a,b}}[\phi(X)\phi(Y)]=0.$ This implies that the four functions $1, \phi(X),\phi(Y),\phi(X)\phi(Y)$ form an orthonormal bases w.r.t. the inner-product defined by the measure $\mu_{0,a,b}$. Consequently, these four functions are linearly independent. Therefore, we can write any function $f:\{\pm1\}^2\to\R$ by the following decomposition: 
$$ f(x,y) = f_0+ \alpha\phi(x)+\beta\phi(y)+\gamma\phi(x)\cdot\phi(y).$$

Recall that $L_Xf = f - \E[f|Y],$ and $L_Yf = f - \E[f|X].$ It is not hard to see that the constant term of $f$ will vanish when applying $L_X,L_Y$ operators. So, without loss of generality we henceforth drop the constant term of $f$.

\begin{lemma}[Explicit description of $A$]
\label{lem:explicit-A} For a function $f:\{\pm1\}^2 \to \R$ with $f=f_0 + \alpha\phi(x)+\beta\phi(y)+\gamma\phi(x)\phi(y)$,
let $A \in \R^{3 \times 3}$ be such that 
\begin{align}\label{eq:local-ineqa} &\begin{bmatrix}\alpha& \beta & \gamma\end{bmatrix} A \begin{bmatrix}\alpha& \beta & \gamma\end{bmatrix}^\top = \E\big[(L_X f)(L_Yf)\big] + J\cdot \E\big[M(L_Xf)(L_Yf)\big] \notag\\&\hspace{4em}+  \frac{(\lambda J)^2}{2}\cdot \E\big[(L_Xf)^2
    + (L_Yf)^2\big] + \frac{dJ^2}{2}\E\big[\sol{\bX}(y)^2\cdot L_Xf^2 + \sol{\bY}(x)^2\cdot  L_Yf^2\big]. 
    \end{align}
    Suppose also that $|J|\leq \tfrac{1}{2}$. Then we can write,  $A = \diag(J, J, 1)(A_0 + E)\diag(J,J,1)$
    where $\|E\|_{\op}\leq O(|J|)$ and
\begin{equation}
A_0 =
\begin{pmatrix}
 \frac{1}{2}\bigl(\lambda^2+d\sol{X}_0^2\bigr)
&
 -\frac{1}{2}\bigl(s^2+2s-1\bigr)
 \sol{X}_0\sol{Y}_0 \sigma_{X_0}\sigma_{Y_0}
&
 -\frac{1}{2}(1-s)\sol{X}_0\sigma_{Y_0}
\\[0.4em]
 -\frac{1}{2}\bigl(s^2+2s-1\bigr)
 \sol{X}_0\sol{Y}_0 \sigma_{X_0}\sigma_{Y_0}
&
 \frac{1}{2}\bigl(\lambda^2+d\sol{Y}_0^2\bigr)
&
 -\frac{1}{2}(1-s)\sol{Y}_0\sigma_{X_0}
\\[0.4em]
 -\frac{1}{2}(1-s)\sol{X}_0\sigma_{Y_0}
&
 -\frac{1}{2}(1-s)\sol{Y}_0\sigma_{X_0}
&
 1
\end{pmatrix}.
\end{equation}
\end{lemma}

We defer the proof of Lemma~\ref{lem:explicit-A} to the end of this section. First we give a proof Theorem \ref{lem:main} assuming the above lemma.

\begin{proof}[Proof of Theorem~\ref{lem:main}]
    We will see that entries of $A_0$ are the leading order terms in the Taylor expansions of the entries of $A$.

Since $\|E \|_{\op} \leq O(|J|)$, to prove that $A$ is PSD, we prove that there exists an absolute constant $c>0$ as a function of $d,s, \lambda$ such that we have $A_0 \succeq cI$. Having that, the lemma's statement follows by choosing $j(d,\lambda,s)=O(c)$ such that for $J\leq j(d,\lambda,s)$ we have $\norm{E}_{\op}\leq c$. 

As a first step, we will show that 
\begin{equation}\label{eq:A0-psd}
    [\alpha\quad\beta\quad\gamma]\,
A_0\,
[\alpha\quad\beta\quad\gamma]^\top \geq 
\frac{\lambda^2}{2}
(\alpha^2+\beta^2) + \left( \gamma - \frac{1-s}{2} \big( \sol{X}_0 \sigma_{Y_0} \alpha + \sol{Y}_0 \sigma_{X_0} \beta \big) \right)^2.
\end{equation}
    Assuming this, we note that 
\begin{align*}
    \alpha^2 + \beta^2 + \gamma^2 &\leq \frac{2}{\lambda^2} \cdot \left( \frac{\lambda^2}{2}(\alpha^2 +\beta^2) +  \left(
  \gamma
  -\frac{1-s}{2}
   \bigl(
     \sol{X}_0\sigma_{Y_0}\alpha
     +\sol{Y}_0\sigma_{X_0}\beta
   \bigr)
 \right)^2\right) \\&+ \left( 1 \cdot \left( \gamma
  -\frac{1-s}{2}
   \bigl(
     \sol{X}_0\sigma_{Y_0}\alpha
     +\sol{Y}_0\sigma_{X_0}\beta\right) +  \frac{(1-s)\sol{X}_0 \sigma_{Y_0}}{2} \cdot \alpha  + \frac{(1-s)\sol{Y}_0 \sigma_{X_0}}{2} \cdot \beta \right)^2 \\
     &\underset{\substack{\text{Cauchy-Schwarz} \\ \sigma_{X_0}, \sigma_{Y_0}, \sol{X}_0, \sol{Y}_0 \leq 1}}{\leq} \frac{2}{\lambda^2} \cdot \left( \frac{\lambda^2}{2}(\alpha^2 +\beta^2) +  \left(
  \gamma
  -\frac{1-s}{2}
   \bigl(
     \sol{X}_0\sigma_{Y_0}\alpha
     +\sol{Y}_0\sigma_{X_0}\beta
   \bigr)
 \right)^2\right)\\
     &+\left( 1 + \frac{(1-s)^2}{2\lambda^2} + \frac{(1-s)^2}{2\lambda^2} \right) \left( \left(
  \gamma
  -\frac{1-s}{2}
   \bigl(
     \sol{X}_0\sigma_{Y_0}\alpha
     +\sol{Y}_0\sigma_{X_0}\beta
   \bigr)
 \right)^2 + \frac{\lambda^2}{2}(\alpha^2 +\beta^2)  \right)  \\
 &= \left( \frac{2}{\lambda^2} + 1 + \frac{(1-s)^2}{\lambda^2} \right) \left( \left(
  \gamma
  -\frac{1-s}{2}
   \bigl(
     \sol{X}_0\sigma_{Y_0}\alpha
     +\sol{Y}_0\sigma_{X_0}\beta
   \bigr)
 \right)^2 + \frac{\lambda^2}{2}(\alpha^2 +\beta^2)  \right)\\
 &\leq \left( \frac{2}{\lambda^2} + 1 + \frac{(1-s)^2}{\lambda^2} \right)    [\alpha\quad\beta\quad\gamma]\,
\sol{A}_0\,
[\alpha\quad\beta\quad\gamma]^\top
\end{align*}
Consequently, this shows that for $c = \left( \frac{2}{\lambda^2} +1 + \frac{(1-s)^2}{\lambda^2} \right)^{-1} $ we have $A_0\succeq cI$. 

In the rest of the proof we show \eqref{eq:A0-psd}. Expanding the quadratic form gives
\begin{align*}
&[\alpha\quad\beta\quad\gamma]\,
A_0 \,
[\alpha\quad\beta\quad\gamma]^\top\\
&=
\frac{\lambda^2}{2}(\alpha^2+\beta^2)
+\frac{d}{2}
 \bigl(
   \sol{X}_0^2\alpha^2
   +\sol{Y}_0^2\beta^2
 \bigr)
-(s^2+2s-1)
 \sol{X}_0\sol{Y}_0
 \sigma_{X_0}\sigma_{Y_0}\,\alpha\beta -(1-s)\gamma
 \bigl(
   \sol{X}_0\sigma_{Y_0}\alpha
   +\sol{Y}_0\sigma_{X_0}\beta
 \bigr)
+\gamma^2 \\
&\underset{\sol{X}_0^2=\sol{X}_0^2\sol{Y}_0^2
+\sol{X}_0^2\sigma_{Y_0}^2} 
{=}
\frac{\lambda^2
+d\sol{X}_0^2\sol{Y}_0^2}{2}
(\alpha^2+\beta^2)
+\frac{d+(s^2+2s-1)}{4}
\bigl(
  \sol{X}_0\sigma_{Y_0}\alpha
  -\sol{Y}_0\sigma_{X_0}\beta
\bigr)^2
\\
&\quad\quad\quad
+\frac{d-(s^2+2s-1)-(1-s)^2}{4}
\bigl(
  \sol{X}_0\sigma_{Y_0}\alpha
  +\sol{Y}_0\sigma_{X_0}\beta
\bigr)^2+\left(
  \gamma
  -\frac{1-s}{2}
   \bigl(
     \sol{X}_0\sigma_{Y_0}\alpha
     +\sol{Y}_0\sigma_{X_0}\beta
   \bigr)
 \right)^2 .
\end{align*}

To get \eqref{eq:A0-psd} we just need to choose $s,d$ such that the coefficients of the second and the third terms in the RHS are non-negative. This holds when $d \geq \max \{ -(s^2 + 2s -1), 2s^2 \}$, which is exactly \eqref{eq:d-s-constraint}. This completes the proof \eqref{eq:A0-psd} and the theorem.
\end{proof}

In the remainder of this section, we prove Lemma~\ref{lem:explicit-A}. 

To prove the lemma we first need to calculate the action of the $L_X,L_Y$ operators on $\phi(X),\phi(Y),$ and $\phi(X)\phi(Y).$
First, notice
\begin{equation}\label{eq:LYphix}
    L_Y\phi(X) = \phi(X) - \E[\phi(X) \, | \, x] = 0 \quad\text{and}\quad L_X\phi(Y)=0.
\end{equation}
Second,
\begin{equation}\label{eq:LXphix}
    L_X \phi(X) = \phi(X) - \E[\phi(X) \, | \, y] = \frac{X - \sol{X}_0}{\sigma_{X_0}} - \frac{\sol{\bX}(y) - \sol{X}_0}{\sigma_{X_0}} = \frac{X - \sol{\bX}(y)}{\sigma_{X_0}} =: R_X(y),\quad L_Y\phi(Y)=:R_Y(x).
\end{equation}

Lastly,
\begin{equation}\label{eq:LXphixphiy}
    L_X(\phi(X)\phi(Y)) = \phi(X)\phi(Y)-\E[\phi(X) \phi(Y)|y] = \phi(Y) L_X \phi(X) =\phi(Y) \cdot R_X, 
\end{equation}
Similarly, we have 
$L_Y(\phi(X)\phi(Y))= \phi(X) \cdot R_Y$.

By linearity for $f = \alpha \phi(X) + \beta \phi(Y) + \gamma \phi(X) \phi(Y)$, we have $L_X f(x,y) = (\alpha + \gamma \phi(Y)) \cdot R_X$, $L_Y f(x,y) = (\beta + \gamma \phi(X)) \cdot R_Y$ and 
\begin{equation}\label{eq:LXfLYf}
    (L_Xf)(L_Yf) = R_XR_Y(\alpha + \gamma \phi(Y)) (\beta + \gamma \phi(X)) = R_XR_Y(\alpha \beta + \alpha \gamma \phi(X) + \beta \gamma \phi(Y)+\gamma^2 \phi(X) \phi(Y)).
\end{equation}

A necessary step to prove Theorem~\ref{lem:main} is to calculate the Taylor expansion of several polynomials of $M,\phi(.),R_X,R_Y$. We do that in the following lemma and then we prove the theorem.

\begin{restatable}[Taylor expansion lemma]{lemma}{Taylor}
\label{lem:consolidated}
Recall $\sol{J}=\tanh(J)$. We have,
    \begin{enumerate}
        \item [(i)] $\E\big[\phi(X)^2\big]= 1 - 2 \sol{J}\sol{X}_0 \sol{Y}_0 \pm  O(\sol{J}^2)$ and $\E\big[\phi(Y)^2\big]= 1 - 2 \sol{J}\sol{X}_0 \sol{Y}_0 \pm O(\sol{J}^2)$.
        \item [(ii)] $\E\big[ R_X^2\big] = 1 \pm O\big(\big|\sol{J}\big|\big)$ and $\E\big[ R_Y ^2\big] = 1 \pm O\big(\big|\sol{J}\big|\big)$.
        \item [(iii)] $\E[R_XR_Y] = -\sigma_{X_0}\sigma_{Y_0} \sol{J} + 2 \sol{X}_0\sol{Y}_0 \sigma_{X_0}\sigma_{Y_0}\sol{J}^2 \pm O\big(\big|\sol{J}\big|^3 \big)$.
        \item [(iv)] $\mathbb{E}[R_X R_Y\varphi(X)] = \pm O\big(\sol{J}^2\big)$ and  $\mathbb{E}[R_X R_Y \varphi(Y)] =\pm O\big(\sol{J}^2 \big) $.
        \item [(v)] $\E[R_XR_Y \varphi(X) \varphi(Y)]  = 1\pm O\big(\sol{J}^2 \big)$.
        \item [(vi)] $\E[MR_XR_Y] = \sigma_{X_0} \sigma_{Y_0} - (s+1)^2 \sol{X}_0 \sol{Y}_0  \sigma_{X_0} \sigma_{Y_0} \sol{J} \pm  O\big(\sol{J}^2\big)$.
        \item [(vii)] $\E[MR_XR_Y \phi(X)] = (s-1)\sol{X}_0\sigma_{Y_0} \pm O\big(\big|\sol{J}\big|\big)$ and $\E[MR_XR_Y \phi(Y)] = (s-1) \sol{Y}_0\sigma_{X_0} \pm O\big(\big| \sol{J} \big|\big)$.
        \item [(viii)] $\E[MR_XR_Y \phi(X) \phi(Y)] = (s-1)^2 \sol{X}_0\sol{Y}_0 \pm O\big(\big|\sol{J}\big|\big)$. 
    \end{enumerate}
    All the $\E$ here are with respect to $\mu_{J,a,b}$. 
\end{restatable}

Equipped with these, we are ready to prove Lemma~\ref{lem:explicit-A}.

\begin{proof}[Proof of Lemma~\ref{lem:explicit-A}] We make a case for every entry of $A$ and use the fact that $\sol{J} = J + O(J^3)$.

\noindent{{$\mathbf{A_{11}, A_{22}}$}.} Plugging in $f = \phi(X)$ (corresponding to the vector $[1 \quad 0 \quad 0]$) into the RHS of \eqref{eq:local-ineqa}. 

Consequently, substituting $\phi(X)$ into the LHS of the theorem gives 
\[ A_{11} \underset{\eqref{eq:LXphix}}{=} \frac{(\lambda J)^2}{2} \cdot \E[R_X^2] + \frac{dJ^2}{2} \cdot \E\big[ \sol{\bX}(Y)^2 R_X^2\big], \]
where we used \eqref{eq:LYphix} so only the last two terms remain.
We bound the first term using Lemma~\ref{lem:consolidated}(i). For the second term, we combine Lemmas~\ref{lem:bfx} and ~\ref{lem:consolidated}(i) to get 
\[ \E\big[ \sol{\bX}(Y)^2 R_X^2\big] = \E\big[\sol{\bX}(Y)^2 \E \big[  R_X^2\, \big | \, Y\big]\big] = \big(\sol{X}_0^2 \pm  O\big(\big|\sol{J}\big|\big)\big)\big(1 \pm O\big(\big|\sol{J}\big|\big)\big) = \sol{X}_0^2 \pm O(|J|).\]
All in all, we get 
\[ A_{11} = \frac{J^2}{2}\,\big( \lambda^2 + d \sol{X}_0^2 \big) \pm O(|J|^3).\]

Similarly, we have $A_{22} = \frac{J^2}{2}\big(\lambda^2 + d \sol{Y}_0^2 \big) \pm O(|J|^3).$

\noindent{$\mathbf{A_{33}}$.} Plugging in $f = \phi(X)\phi(Y)$ (corresponding to the vector $[0 \quad 0 \quad 1]$) into the RHS of \eqref{eq:local-ineqa}, gives
\begin{align*}
    A_{33} &\underset{\eqref{eq:LXphixphiy}}{=} \E[R_XR_Y \phi(X)\phi(Y)] + J \E[MR_XR_Y\phi(X)\phi(Y)] + \frac{(\lambda J)^2}{2} \E \big[ R_X^2 \phi(Y)^2 + R_Y^2 \phi(X)^2 \big] \\&\hspace{4em}+ \frac{dJ^2}{2}\E\big[\sol{\bX}(Y)^2 R_X^2 \phi(Y)^2 + \sol{\bY}(X)^2R_Y^2 \phi(X)^2 \big] \underset{\text{Lemma~\ref{lem:consolidated}(iv)}}{=} 1 \pm O(|J|).
\end{align*}
To see the that third and fourth term are bounded by $O(J^2)$ we note, for example, that
\begin{align*}
    \E_{\mu_{J,a,b}} \big[ R_X^2 \phi(Y)^2 \big] &\underset{(\dagger)}{\leq}  \E_{\mu_{J,a,b}}[\phi(X)^2 \phi(Y)^2] \\
    &\underset{\text{Lemma~\ref{lem:Radon-Nikodym}}}{\leq} \big(1 + \sol{J}(xy - \sol{X}_0\sol{Y}_0 \big) + O\big( \sol{J}^2 \big) \big) \E_{\mu_0,a,b}[ \phi(X)^2 \phi(Y)^2] \\
    &\underset{\text{independence}}{\leq} \big(1 + 2 \sol{J} + O\big(\sol{J}^2\big) \big)\,\E_{\mu_0,a,b}[\phi(X)^2] \, \E_{\mu_0,a,b}[\phi(Y)^2]  \leq O(1),
\end{align*}
where $(\dagger)$ follows because
\begin{equation}\label{eq:var-cond-Y}
    \E_{\mu_{J,a,b}}[R_X^2 \, | \, Y] = \Var_{\mu_{J,a,b}}(\phi(X) \, | \, Y)\leq \E_{\mu_{J,a,b}}[\phi(X)^2 \, | \, Y]
\end{equation}
and so the tower property implies that 
\begin{equation}\label{eq:phiY2-RX2}
    \E_{\mu_{J,a,b}}[\phi(Y)^2 R_X^2] = \E_{\mu_{J,a,b}}[\phi(Y)^2 \E_{\mu_{J,a,b}}[R_X^2 \, | \, Y]] \underset{\eqref{eq:var-cond-Y}}{\leq} \E_{\mu_{J,a,b}}[\phi(Y)^2 \E_{\mu_{J,a,b}}[\phi(X)^2 \, | \, Y]] = \E_{\mu_{J,a,b}}[\phi(Y)^2 \phi(X)^2].  
\end{equation}
The above shows the third term is bounded. To show the boundedness of the fourth term we combine the above with $\sol{\bX}(y)^2 \leq 1$, for $y \in \{\pm 1\}$. 

\noindent{$\mathbf{A_{12}}$.} 
Plugging in $f = \alpha \phi(X) + \beta \phi(Y)$ into the RHS of \eqref{eq:local-ineqa}, we can read off $A_{12}$ as half of the coefficient of $\alpha\beta$,
 
\begin{align*}
    A_{12} &\underset{\eqref{eq:LXfLYf}}{=} \frac{1}{2} \E[R_XR_Y] + \frac{J}{2}\E[MR_XR_Y] \\
    &\underset{\text{Lemma~\ref{lem:consolidated}(ii),(iv)}}{=}-\frac{1}{2}\sigma_{X_0}\sigma_{Y_0} \sol{J}+ \sol{X}_0\sol{Y}_0\sigma_{X_0} \sigma_{Y_0} \sol{J}^2  + \frac{\sol{J}}{2}\big(\sigma_{X_0}\sigma_{Y_0}- (s+1)^2 \sol{X}_0 \sol{Y}_0 \sigma_{X_0} \sigma_{Y_0} \sol{J} \big) + O\big(\big|\sol{J}\big|^3\big) \\
    &= -\frac{J^2}{2}(s^2 + 2s -1)\sol{X}_0\sol{Y}_0\sigma_{X_0} \sigma_{Y_0} \pm O(|J|^3).
\end{align*}

\noindent{$\mathbf{A_{13}, A_{23}}$.} 
Plugging in $f = \alpha \phi(X) + \gamma\phi(X) \phi(Y)$ into the RHS of \eqref{eq:local-ineqa}, we can read off $A_{13}$ as half of the coefficient of $\alpha\gamma$,

\begin{align*}
     A_{13} &= \frac{1}{2}\E[R_XR_Y\phi(X)] + \frac{J}{2}\E[MR_XR_Y\phi(X)] + (\lambda J)^2 \E[\phi(Y) R_Y^2] + dJ^2\E\big[\sol{\bX}(y)^2 R_X^2 \phi(Y) \big] \\
     &\hspace{-1.5em}\underset{\text{Claim~\ref{lem:consolidated}(iii), (vi)}}{\leq} \frac{J}{2} (s-1)\sol{X}_0\sigma_{Y_0} + J^2 \left( \lambda^2 \cdot \E[\phi(Y) R_X^2 ] + d \E[\sol{\bX}(y)^2 R_X^2 \phi(Y)] \right) \leq J (s-1)\sol{X}_0\sigma_{Y_0} \pm O(J^2).
\end{align*}
Here the final inequality holds because similarly to \eqref{eq:phiY2-RX2}, 
we can show that
\begin{equation}\label{eq:phiY-RX}
    \E_{\mu_{J,a,b}}[|\phi(Y)| R_X^2] = \E_{\mu_{J,a,b}}[|\phi (Y)|\E_{\mu_{J,a,b}}[R_X^2 \, | \, Y]] \underset{\eqref{eq:var-cond-Y}}{\leq} \E_{\mu_{J,a,b}}[|\phi(Y)| \E_{\mu_{J,a,b}}[\phi(X)^2 \, | \, Y]] = \E_{\mu_{J,a,b}}[|\phi(Y)|\phi(X)^2], 
\end{equation}
which in turn shows that
\begin{align*}
    \E_{\mu_{J,a,b}}[\phi(Y) R_X^2] &\leq \E_{\mu_{J,a,b}}[|\phi(Y)| R_X^2] \underset{\substack{\text{Lemma~\ref{lem:Radon-Nikodym}}\\ \eqref{eq:phiY-RX}}}{\leq} \big(1 + 2 \sol{J} + O\big(\sol{J}^2\big) \big) \E_{\mu_{0,a,b}}[|\phi(Y)| \phi(X)^2] \\
    &\leq \big(1 + 2 \sol{J} + O\big(\sol{J}^2\big) \big) \E_{\mu_{0,a,b}}[|\phi(Y)|] \E_{\mu_{0,a,b}}[\phi(X)^2] \\
    &\underset{\text{Cauchy-Schwarz}}{\leq}\big(1 + 2 \sol{J} + O\big(\sol{J}^2\big) \big) \big(\E_{\mu_{0,a,b}}[\phi(Y)^2] \big)^{1/2}\E_{\mu_{0,a,b}}[\phi(X)^2] = O(1).
\end{align*}
This shows that $ \lambda^2 \cdot \E[\phi(Y) R_X^2 ]$ is bounded. To show that $d \E[\sol{\bX}(y)^2 R_X^2 \phi(Y)]$ is bounded we combine the above with the observation that $ \sol{\bX}(y)^2  \leq 1$ for $y \in \{ \pm 1 \}$. 

Similarly, we have  $A_{23} = \frac{J}{2} (s-1)\sol{Y}_0\sigma_{X_0} \pm O(J^2).$ 
To finish the proof we notice that after subtracting $A_0$, every entry of error matrix $E$ is bounded by $O(|J|)$. So, $\|E\|_{\op}\leq O(|J|)$. This completes the proof of Lemma~\ref{lem:explicit-A}.
\end{proof}
\section{Proof of Lemma~\ref{lem:consolidated}}
In this section, we prove Lemma~\ref{lem:consolidated}.
Before proving the statement we record a few equations which will be useful:
\begin{equation}\label{eq:Xphi}
    \E_{\mu_{0,a,b}}[X \phi(X)]= \frac{\E_{\mu_{0,a,b}}[X^2]- \sol{X}_0\E_{\mu_{0,a,b}}[X]}{\sigma_{X_0}} = \frac{1 - \sol{X}_0^2}{\sigma_{X_0}} = \sigma_{X_0}
\end{equation}
and we note that $(X- \sol{X}_0)^2 =   (1- \sol{X}_0^2) - 2X\sol{X}_0 + 2\sol{X}_0^2= \sigma_{X_0}^2-2\sol{X}_0(X - \sol{X}_0)$ so that
\begin{align}\label{eq:Xphisquare}
     \E_{\mu_{0,a,b}}[X \phi(X)^2]&= \E_{\mu_{0,a,b}}\bigg[X \cdot  \frac{\sigma_{X_0}^2 - 2 \sol{X}_0\big(X - \sol{X}_0\big)}{\sigma_{X_0}^2}\bigg] = \E_{\mu_{0,a,b}}\bigg[X \cdot \bigg( 1 - \frac{2 \sol{X}_0 \phi(X)}{\sigma_{X_0}}\bigg)\bigg] \notag \\&= \sol{X}_0- \frac{2\sol{X}_0}{\sigma_{X_0}} \cdot \E_{\mu_{0,a,b}}[X \phi(X)] \underset{\eqref{eq:Xphi}}{=}  \sol{X}_0- \frac{2\sol{X}_0}{\sigma_{X_0}} \cdot \sigma_{X_0}= - \sol{X}_0.
\end{align}

Throughout we assume that $\big|\sol{J} \big|\leq \tfrac{1}{2}$, and we recall $ |\sol{X}_0|, |\sol{Y}_0| \leq 1$ We have 
\begin{align}\label{eq:explicit-R}
    R_X(y) &= \frac{X - \sol{\bX}(y)}{\sigma_{X_0}} \underset{\text{Lemma}~\ref{lem:mu}}{=} \frac{\big(X - \sol{X}_0\big) \big(1 - \sol{J}Xy\big)}{\sigma_{X_0}\big(1+ \sol{J} \sol{X}_0 y)} \notag \\ &\hspace{2em} = \phi(X) \cdot \frac{1 - \sol{J}Xy}{1+ \sol{J}\sol{X}_0y}\leq O(\phi(X)) \quad\text{and}\quad R_Y(x) \leq O(\phi(Y)). 
\end{align}

Furthermore, 
\begin{equation}\label{eq:explicit-dotR}
    |\dot{R_X}(y)| = \bigg|\phi(X) \cdot \frac{d}{d\tilde{J}} \frac{1 - \sol{J}Xy}{1+ \sol{J}\sol{X}_0y} \bigg|= \bigg|\phi(X) \cdot \frac{y(X+\sol{X}_0)}{\big(1 + \sol{J} \sol{X}_0y \big)^2}\bigg| \leq O(\phi(X)) \quad\text{and}\quad |\dot{R_Y}(x)| \leq  O(\phi(Y)). 
\end{equation}
We also have 
\begin{equation}\label{eq:explicit-ddotR}
    |\ddot{R_X}(y)| =  \bigg|\phi(X) \cdot \frac{d}{d\sol{J}}\frac{y(X+\sol{X}_0)}{\big(1 + \sol{J} \sol{X}_0y \big)^2}\bigg| = \bigg|\phi(X) \cdot \frac{2 \sol{X}_0(X+ \sol{X}_0)}{(1+ \sol{J} \sol{X}_0y)^3}\bigg|\leq O(\phi(X)) \quad\text{and}\quad |\dot{R_Y}(x)| \leq  O(\phi(Y)). 
\end{equation}

By Lemma~\ref{lem:bfx}, we also have 
\begin{equation}\label{eq:taylor-R}
    R_X(y) \underset{\eqref{eq:LXphix}}{=} \phi(X) - \sol{J} \sigma_{X_0}y \pm O\big(\sol{J}^2 \big) \quad \text{and} \quad R_Y(x) \underset{\eqref{eq:LXphix}}{=} \phi(Y) - \sol{J} \sigma_{Y_0}x \pm O\big( \sol{J}^2\big).
\end{equation}

Next, we write,
\begin{align}\label{eq:taylor-M}
    M(\sol{J},x,y) &\underset{\text{Theorem }\ref{lem:main}}{=}
    (x+s\sol{X})(y+s\sol{Y})\notag\\&\underset{\text{Lemma~\ref{lem:bfx}}, \; s\leq 1}{=} \big(x + s\big(\sol{X}_0 + \sol{J} \sigma_{X_0}^2 y \big) \big) \big( y+ s\big( \sol{Y}_0 + \sol{J}\sigma_{Y_0}^2 x \big) \big)  \pm O\big(\sol{J}^2 \big) \notag \\
    &= \big( x + s \sol{X}_0 \big) \big( y + s \sol{Y}_0 \big) + \sol{J}\big( \big(x + s \sol{X}_0 \big) s \sigma_{Y_0}^2 x + \big(y + s \sol{Y}_0 \big) s \sigma_{X_0}^2 y \big) \pm O\big( \sol{J}^2 \big),
\end{align}
and we note that 
\begin{align}\label{eq:dotM}
    |\dot{M}(J,x,y)| &\leq s\big|\dot{\sol{\bX}}(y)\big|\cdot \big|y + s \sol{\bY}(x)\big| + s \big|\dot{\sol{\bY}}(x) \big| \cdot\big|x + s \sol{\bX}(y) \big| \notag\\
    &\hspace{-1em}\underset{\substack{x,y,s,\sol{\bX},\sol{\bY} \leq 1 \\ \text{Lemma~\ref{lem:mu}}}}{\leq} O\Bigg( \frac{y\big(1 - \sol{X}_0^2 \big)}{(1 + \sol{X}_0\sol{J}y)^2} \Bigg) + O\Bigg(\frac{x\big(1 - \sol{Y}_0^2 \big)}{(1 + \sol{Y}_0\sol{J}x)^2} \Bigg) \underset{\substack{x,y,\sol{X}_0,\sol{Y}_0 \leq 1 \\ \big|\sol{J}\big| \leq \tfrac{1}{2}}}{\leq} O(1).
\end{align}

Finally, we have 
\begin{align}\label{eq:ddotM}
    |\ddot{M}(J,x,y)| &\leq s\big| \ddot{\sol{\bX}}(y)\big|\cdot \big|y + s \sol{\bY}(x) \big| + s \big|\ddot{\sol{\bY}}(x) \big|\cdot \big| x +s \sol{\bX}(y) \big| + 2s^2 \big| \dot{\sol{\bX}}(y) \big| \big| \dot{\sol{\bY}}(x) \big| \notag\\
    &\hspace{-1em}\underset{\substack{x,y,s,\sol{\bX},\sol{\bY} \leq 1 \\ \text{Lemma~\ref{lem:mu}}}}{\leq} O\Bigg( \frac{\sol{X}_0(1 - \sol{X}_0^2)}{(1 + \sol{X}_0\sol{J}y)^3}  \Bigg) +  O\Bigg( \frac{\sol{Y}_0(1 - \sol{Y}_0^2)}{(1 + \sol{Y}_0\sol{J}x)^3}  \Bigg) + O(1)\underset{\substack{x,y,\sol{X}_0,\sol{Y}_0 \leq 1 \\ \big|\sol{J}\big| \leq \tfrac{1}{2}}}{\leq}  O(1).
\end{align}

\begin{lemma}\label{lem:bounded-second-der}
Let $k \geq 1$ be an integer and suppose Let $F\colon \mathbb{R} \times \{ \pm 1 \}^2 \to \mathbb{R}$ is such that 
\[ F(\sol{J},x,y) =  R_X({\sol{J}},y)^{n_{R_X}} R_Y({\sol{J}},x)^{n_{R_Y}} \phi(X)^{n_X} \phi(Y)^{n_Y} M(\sol{J},x,y)^{n_M} \]
such that $n_{R_X}+ n_X \leq 2$, $n_{R_Y} + n_Y \leq 2$ and $n_M \leq 1$. 
Then there is a constant $C$ such that
\begin{equation}\label{eq:derivative-L1-bound}
\sup_{|{\sol{J}}|\leq 1/2} \E_{\mu_{0,a,b}}|F|, \, \sup_{|{\sol{J}}|\le1/2}\E_{\mu_{0,a,b}}\bigg|\frac{\partial F}{\partial \sol{J}}\bigg|, \, \sup_{|{\sol{J}}|\le1/2}\E_{\mu_{0,a,b}}\bigg|\frac{\partial^2 F}{\partial \sol{J}^2}\bigg|
 \leq C.
\end{equation}
\end{lemma}

\begin{proof}
By \eqref{eq:explicit-R}, \eqref{eq:explicit-dotR}, \eqref{eq:explicit-ddotR} and $\big|\sol{J} \big| \leq \tfrac{1}{2}$ we write
\begin{align*}
    F \leq O(1) |\phi(X)|^{n_X+n_{R_X}}|\phi(Y)|^{n_Y+n_{R_Y}}|M|^{n_M} \underset{\eqref{eq:taylor-M}}{\leq} O\big(|\phi(X)|^{n_X+n_{R_X}}|\phi(Y)|^{n_Y+n_{R_Y}}\big),
\end{align*}
and
\begin{equation*}
    \big| \dot{F} \big|, \big|\ddot{F}\big|\underset{\eqref{eq:taylor-M}, \eqref{eq:dotM}, \eqref{eq:ddotM}}{\leq} O(1) \big(1+|\phi(X)|^{n_X+n_{R_X}} \big) \big(1+ |\phi(Y)|^{n_Y+n_{R_Y}} \big). 
\end{equation*}

By assumption, $n_X + n_{R_X}, n_Y + n_{R_Y} \in \{ 0, 1 , 2 \}$. It remains to show that we can upper bound $\E_{\mu_{0,a,b}}[|\phi(X)|^{m_X} |\phi(Y)|^{m_Y}]$ by a constant for $m_X \in \{0,1,2\}$ and $m_Y \in \{0,1,2\}$. To that end, note
\begin{equation}\label{eq:phi-properties}
 \E_{\mu_{0,a,b}}|\phi(X)| \underset{\text{Cauchy-Schwarz}}{\leq} (\E[\phi(X)^2])^{1/2} = 1 \quad\text{and} \quad\E_{\mu_{0,a,b}}|\phi(Y)|\leq 1.
\end{equation}
Consequently, it follows that
\[ \E_{\mu_{0,a,b}}[|\phi(X)|^{m_X} |\phi(Y)|^{m_Y}] \underset{\text{independence}}{\leq} \E_{\mu_{0,a,b}}[|\phi(X)|^{m_X}]\E_{\mu_{0,a,b}}[ |\phi(Y)|^{m_Y}] \underset{\eqref{eq:phi-properties}}{\leq} 1. \qedhere\]
\end{proof}

\Taylor*
\begin{proof}
\textbf{(i).} Let $F = \phi(X)^2$. By Lemma~\ref{lem:bounded-second-der}, the conditions  \eqref{eq:F-conditions}  ($\E[F(0)] \leq \E[|F(0)|] \leq O(1)$, $\E[|\dot{F}(0)|] \leq O(1) $  and $ \E\big[\big|\ddot{F}\big(\sol{J}^* \big)\big|\big] \leq O(1) $)  hold for this choice of $F$. So we apply Lemma~\ref{lem:master} to $F$. 
We observe that $\frac{d}{d\sol{J}} \phi(X)^2 \bigr|_{\sol{J}= 0} = 0$ since $\phi(X)$ does not depend on $\sol{J}$.  
Consequently, 
\begin{align*}
    \E_{\mu_{J,a,b}}\big[\phi(X)^2\big] &= \E_{\mu_{0,a,b}}[\phi(X)^2] + \sol{J}\mathbb{E}_{\mu_{0,a,b}}[(XY - \sol{X}_0 \sol{Y}_0)\phi(X)^2] \pm O\big(\sol{J}^2\big) \\&= 1+ \sol{J}\big(\E_{\mu_{0,a,b}}[X\phi(X)^2] \E_{\mu_{0,a,b}}[Y] - \sol{X_0}\sol{Y_0})\big) \pm O \big(\sol{J}^2 \big) =1 - 2 \sol{X}_0\sol{Y}_0\sol{J} \pm O\big( \sol{J}^2 \big).
\end{align*}
The second assertion follows by exchanging the roles of $X$ and $Y$.

\textbf{(ii).} Let $F = R_X^2$. By Lemma~\ref{lem:bounded-second-der}, the conditions \eqref{eq:F-conditions} hold for this choice of $F$, so we apply Lemma~\ref{lem:master} to $F$.  We first use \eqref{eq:taylor-R} to write 
\[ R_X(y)^2 = \phi(X)^2 -2 \sol{J}\sigma_{X_0} \phi(X) y \pm O(\sol{J}^2),\] 
where there is an absolute constant in the $O(\cdot)$. This implies that $\E_{\mu_{0,a,b}}[F(0)] = \E_{\mu_{0,a,b}}[\phi(X)^2] = 1$. Furthermore,
\begin{align*}
    \E_{\mu_{0,a,b}}[\dot{F}(0)] + \E_{\mu_{0,a,b}}\big[\big(XY - \sol{X}_0\sol{Y}_0 \big) F(0) \big] &\underset{\text{triangle ineq}}{\leq} \E_{\mu_{0,a,b}}[|\dot{F}(0)|] + \E_{\mu_{0,a,b}}\big[\big|\big(XY - \sol{X}_0\sol{Y}_0 \big) F(0) \big| \big] \\
    &\leq \E_{\mu_{0,a,b}}[|\dot{F}(0)|] + 2\E_{\mu_{0,a,b}}[|F(0)|] \underset{\text{Lemma~\ref{lem:bounded-second-der}}}{\leq} O(1).
\end{align*}
Consequently, the linear term in $\sol{J}$ of $R_X^2$ is $O\big(\sol{J} \big)$.

\textbf{(iii).} We have the following explicit expressions:
\[ R_X(y)= \frac{X - \sol{X}(y)}{\sigma_{X_0}} \underset{\text{Lemma~\ref{lem:mu}}}{=} \frac{(X-\sol{X}_0)(1-\sol{J}Xy)}{\sigma_{X_0}\,\Big(1+\sol{J}\sol{X}_0y\Big)} \quad \text{and}\quad R_Y(x)\underset{\text{Lemma~\ref{lem:mu}}}{=} \frac{(Y-\sol{Y}_0)(1-\sol{J}Yx)}{\sigma_{Y_0}\,\Big(1+\sol{J}\sol{Y}_0x\Big)}.\]
It follows that
\begin{equation}\label{eq:muUV}
    \mu(X,Y) R_X(Y)R_Y(X) \underset{\text{Lemma~\ref{lem:mu}}}{=} \frac{\sigma_{X_0}\sigma_{Y_0} (1- \sol{J}^2)}{4(1+\sol{J}\sol{X}_0\sol{Y}_0)} \cdot \frac{XY(1 -\sol{J} XY)}{(1+\sol{J}\sol{X}_0Y)(1+\sol{J}\sol{Y}_0X)}
\end{equation}
    so that 
    \begin{align*}
        \E[R_XR_Y] &= \big(-\sigma_{X_0} \sigma_{Y_0} \sol{J} \big) \cdot \frac{\big(1 - \sol{J}^2)}{ \big(1 - \sol{J}^2 \sol{X}_0^2 \big)\big(1 - \sol{J}^2 \sol{Y}_0^2\big)} \cdot \frac{1- \sol{J}\sol{X}_0\sol{Y}_0}{1+ \sol{J}\sol{X}_0\sol{Y}_0 }  \\
        &= \big(-\sigma_{X_0} \sigma_{Y_0} \sol{J} \big) \cdot \big(1 \pm O\big(\sol{J}^2 \big)\big) \cdot \big (1 - 2 \sol{J}\sol{X}_0\sol{Y}_0 \pm O\big(|\sol{J}|^2\big)\big).
    \end{align*}
Expanding this gives the desired.

\textbf{(iv).} 
Let $F = R_X R_Y \phi(X)$.  By Lemma~\ref{lem:bounded-second-der}, the conditions \eqref{eq:F-conditions} of  Lemma~\ref{lem:master} hold, so we apply Lemma~\ref{lem:master} to $F$. We first write 
\[ R_XR_Y \phi(X)= \phi(X)^2 \phi(Y) - \sol{J} \big( \phi(X)^2 \sigma_{Y_0}X + \phi(X) \phi(Y) \sigma_{X_0} Y \big) \pm O\big(\sol{J}^2 \big). \]

The expectation of the leading term is $\E_{\mu_{0,a,b}}[F(0)] = \E_{\mu_{0,a,b}}[\phi(X)^2] \E_{\mu_{0,a,b}}[\phi(Y)] = 0$. The expectation of the term linear in $\sol{J}$ is 
\begin{align*}
    &-\E_{\mu_{0,a,b}}\big[\phi(X)^2 \sigma_{Y_0}X + \phi(X)\phi(Y) \sigma_{X_0}Y\big] + \E_{\mu_{0,a,b}}\big[(XY - \sol{X}_0\sol{Y}_0) \phi(X)^2 \phi(Y)\big] \\
    &= -\sigma_{Y_0} \E_{\mu_{0,a,b}}[\phi(X)^2X] -\sigma_{X_0} \E_{\mu_{0,a,b}}[\phi(X)]\E_{\mu_{0,a,b}}[\phi(Y)Y] \\ 
    &\hspace{8em}+ \E_{\mu_{0,a,b}}[X\phi(X)^2] \E_{\mu_{0,a,b}}[Y\phi(Y)]-\sol{X}_0\sol{Y}_0\E_{\mu_{0,a,b}}[\phi(X)^2]\E[\phi(Y)] \\
    &\underset{\eqref{eq:Xphi},\, \eqref{eq:Xphisquare}}{=} \sigma_{Y_0}\sol{X}_0 -\sol{X}_0\sigma_{Y_0}  = 0.
\end{align*}

This gives $\E_{\mu_{J,a,b}}[R_XR_Y\phi(X)] = O\big(\sol{J}^2 \big)$. Interchanging the roles of $X$ and $Y$, we also get $\E_{\mu_{J,a,b}}[R_XR_Y\phi(Y)] = O\big(\sol{J}^2 \big)$.

\textbf{(v).} Let $F = R_X R_Y \phi(X)\phi(Y)$.  By Lemma~\ref{lem:bounded-second-der}, the conditions \eqref{eq:F-conditions} of  Lemma~\ref{lem:master} hold, so we apply Lemma~\ref{lem:master} to $F$. We first write 
\[ R_XR_Y \phi(X) \phi(Y)\underset{\eqref{eq:taylor-R}}{=} \phi(X)^2 \phi(Y)^2 - \sol{J} \big( \phi(X)^2 \phi(Y) \sigma_{Y_0}X + \phi(X) \phi(Y)^2 \sigma_{X_0} Y \big) \pm O\big(\sol{J}^2 \big). \]

% \MA{The LHS should be $\E[R_XR_Y\phi(X)]$}

The expected value of the leading term is $\E_{\mu_{0,a,b}}[\phi(X)^2\phi(Y)^2] = \E_{\mu_{0,a,b}}[\phi(X)^2] \E_{\mu_{0,a,b}}[\phi(Y)^2] = 1$. The expected value of the term linear in $\sol{J}$ is 
\begin{align*}
    &-\E_{\mu_{0,a,b}}\big[\phi(X)^2 \phi(Y) \sigma_{Y_0}X + \phi(X)\phi(Y)^2 \sigma_{X_0}Y] + \E_{\mu_{0,a,b}}[(XY - \sol{X}_0\sol{Y}_0) \phi(X)^2 \phi(Y)^2\big] \\
    &= -\sigma_{Y_0} \E_{\mu_{0,a,b}}[X\phi(X)^2] \E_{\mu_{0,a,b}}[\phi(Y)]-\sigma_{X_0} \E_{\mu_{0,a,b}}[Y\phi(Y)^2] \E_{\mu_{0,a,b}}[\phi(X)] \\&\hspace{4em}+ \E_{\mu_{0,a,b}}[X\phi(X)^2]\E_{\mu_{0,a,b}}[Y\phi(Y)^2]- \sol{X}_0\sol{Y}_0\E_{\mu_{0,a,b}}[\phi(X)^2]\E_{\mu_{0,a,b}}[\phi(Y)^2]  \underset{\eqref{eq:Xphisquare}}{=} 0. 
\end{align*}

\textbf{(vi).} Let $F = MR_X R_Y$.  By Lemma~\ref{lem:bounded-second-der}, the conditions \eqref{eq:F-conditions} of  Lemma~\ref{lem:master} hold, so we apply Lemma~\ref{lem:master} to $F$. (Henceforth we no longer write the Taylor expansion of the functions explicitly and instead write each term separately.) First, we compute the constant term.
\begin{align}\label{eq:constant-MRXRY}
    \E_{\mu_{0,a,b}} [M(0)R_X(0)R_Y(0)] &\underset{\eqref{eq:taylor-R},\eqref{eq:taylor-M}}{=} \E_{\mu_{0,a,b}}\big[ \big(X + s \sol{X}_0 \big) \phi(X)\big] \cdot \E_{\mu_{0,a,b}}\big[ \big( Y + s \sol{Y}_0 \big) \phi(Y) \big] \notag \\
    &= \big(\E_{\mu_{0,a,b}}[X \phi(X)]+ s \sol{X}_0\E_{\mu_{0,a,b}}[\phi(X)] \big)\big(\E_{\mu_{0,a,b}}[Y \phi(Y)]+ s \sol{Y}_0\E_{\mu_{0,a,b}}[\phi(Y)] \big) \notag \\
    &\underset{\eqref{eq:Xphi}}{=} \sigma_{X_0}\sigma_{Y_0}.
\end{align}

The linear in $\sol{J}$ term of $ MR_XR_Y$ can be computed as follows: first, we compute using \eqref{eq:taylor-M} and \eqref{eq:taylor-R} the coefficient coming from $\dot{F}_1 (0)$. All the expectations below are with respect to $\mu_{0,a,b}$.
\begin{align}\label{eq:first-MRXRY}
    &\EE \bigg[\frac{d}{d\sol{J}}MR_XR_Y \big|_{\sol{J} = 0} \bigg]= \E\big[ \phi(X) \phi(Y) \big( \big(X + s \sol{X}_0 \big) s \sigma_{Y_0}^2 X + \big(Y + s \sol{Y}_0 \big) s \sigma_{X_0}^2 Y\big)  \notag \\ &\hspace{3em} - \sigma_{X_0} Y \phi(Y) \big(X + s \sol{X}_0 \big) \big(Y + s \sol{Y}_0 \big) - \sigma_{Y_0} X \phi(X) \big(X + s \sol{X}_0 \big) \big( Y + s \sol{Y}_0 \big) \big] \notag \\
    &= s \sigma_{Y_0}^2\E_{\mu_{0,a,b}}[\phi(Y)] \E\big[X\phi(X)\big( X + s \sol{X}_0\big) \big]+s \sigma_{X_0}^2\E[\phi(X)] \E_{\mu_{0,a,b}}\big[Y\phi(Y)\big( Y + s \sol{Y}_0\big) \big] \notag \\
    &\hspace{3em} -\sigma_{X_0}\left( \E [\phi(Y)] s\sol{X}_0 + s^2 \sol{X}_0 \sol{Y}_0 \E[Y \phi(Y)] + \E[\phi(Y)]\E[X] + s \sol{Y}_0\E[Y\phi(Y)] \E[X]\right) \notag \\
     &\hspace{3em} -\sigma_{Y_0}\left( \E [\phi(X)] s\sol{Y}_0 + s^2 \sol{Y}_0 \sol{X}_0 \E[X \phi(X)] + \E[\phi(X)]\E[Y] + s \sol{X}_0\E[X\phi(X)] \E[Y]\right) \notag \\
    &= -2(s^2 + s) \sigma_{X_0} \sigma_{Y_0} \sol{X}_0 \sol{Y}_0.
\end{align}
Second, we compute the coefficient coming from $\big(XY - \sol{X}_0\sol{Y}_0 \big)M(0)R_X(0)R_Y(0)$ using $\E_{\mu_{0,a,b}}[\phi(X)] = \E_{\mu_{0,a,b}}[\phi(Y)] = 0$ as follows 
\begin{align}\label{eq:second-MRXRY}
    \E_{\mu_{0,a,b}} \big[\big(XY - \sol{X}_0 \sol{Y}_0 \big) \big(X + s \sol{X}_0 \big) \phi(X) \big(Y + s \sol{Y}_0 \big)  \phi(Y)\big] &\underset{\eqref{eq:constant-MRXRY}}{} \E_{\mu_{0,a,b}}\big[s^2XY \sol{X}_0 \sol{Y}_0 \phi(X) \phi(Y)\big] - \sol{X}_0\sol{Y}_0 \sigma_{X_0} \sigma_{Y_0} \notag \\
    &= (s^2 -1) \sigma_{X_0} \sigma_{Y_0}\sol{X}_0\sol{Y}_0.
\end{align}

We add \eqref{eq:first-MRXRY} and \eqref{eq:second-MRXRY} to get the term linear in $\sol{J}$, and combine with the constant term \eqref{eq:constant-MRXRY}, which gives the first assertion. The second assertion of (v) follows by interchanging the roles of $X$ and $Y$ above.

\textbf{(vii).} Let $F = MR_XR_Y\phi(X)$. By Lemma~\ref{lem:bounded-second-der}, the conditions \eqref{eq:F-conditions} of  Lemma~\ref{lem:master} hold, so we apply Lemma~\ref{lem:master} to $F$. 
\begin{align}\label{eq:constant-MRXRYphi}
    &\E_{\mu_{0,a,b}} [F(0)] = \E_{\mu_{0,a,b}}\big[ \big(X + s \sol{X}_0 \big) \phi(X)^2\big] \cdot \E_{\mu_{0,a,b}}\big[ \big( Y + s \sol{Y}_0 \big) \phi(Y) \big] \notag \\
    &= \big(\E_{\mu_{0,a,b}}[X \phi(X)^2]+ s \sol{X}_0\E_{\mu_{0,a,b}}[\phi(X)^2] \big)\big(\E_{\mu_{0,a,b}}[Y \phi(Y)]+ s \sol{Y}_0\E_{\mu_{0,a,b}}[\phi(Y)] \big) \notag \\
    &\underset{\eqref{eq:Xphisquare}}{=} (s-1)\sol{X}_0 \cdot \sigma_{Y_0}.
\end{align}
As in (ii), using another application of Lemma~\ref{lem:bounded-second-der}, the linear term in $\sol{J}$ of $ MR_XR_Y\phi(X)$ is $O\big(\sol{J}\big)$. Interchanging the roles of $X$ and $Y$, we obtain the second assertion similarly. 

\textbf{(viii).} Let $F = MR_XR_Y\phi(X)\phi(Y)$. By Lemma~\ref{lem:bounded-second-der}, the conditions \eqref{eq:F-conditions} of  Lemma~\ref{lem:master} hold, so we apply Lemma~\ref{lem:master} to $F$. 
\begin{align}\label{eq:constant-MRXRYphisquare}
    &\E_{\mu_{0,a,b}} [F(0)] \\&= \E_{\mu_{0,a,b}}\big[ \big(X + s \sol{X}_0 \big) \phi(X)^2\big] \cdot \E_{\mu_{0,a,b}}\big[ \big( Y + s \sol{Y}_0 \big) \phi(Y)^2 \big] \notag \\
    &= \big(\E_{\mu_{0,a,b}}[X \phi(X)^2]+ s \sol{X}_0\E_{\mu_{0,a,b}}[\phi(X)^2] \big)\big(\E_{\mu_{0,a,b}}[Y \phi(Y)^2]+ s \sol{Y}_0\E_{\mu_{0,a,b}}[\phi(Y)^2] \big) \notag \\
    &\underset{\eqref{eq:Xphisquare}}{=} (s-1)^2\sol{X}_0 \sol{Y}_0.
\end{align}

As in (ii), using another application of Lemma~\ref{lem:bounded-second-der}, the linear term in $\sol{J}$ of $ MR_XR_Y\phi(X)\phi(Y)$ is $O\big(\sol{J}\big)$. \qedhere
\end{proof}
\bibliographystyle{alpha}
\bibliography{ref.bib}

@inproceedings {AKV24,
    AUTHOR = {Anari, Nima and Koehler, Frederic and Vuong, Thuy-Duong},
     TITLE = {Trickle-down in localization schemes and applications},
 BOOKTITLE = {S{TOC}'24---{P}roceedings of the 56th {A}nnual {ACM}
              {S}ymposium on {T}heory of {C}omputing},
     PAGES = {1094--1105},
 PUBLISHER = {ACM, New York},
      YEAR = {[2024] \copyright 2024},
      ISBN = {979-8-4007-0383-6},
   MRCLASS = {68Q87},
  MRNUMBER = {4764887},
       DOI = {10.1145/3618260.3649622},
       URL = {https://doi.org/10.1145/3618260.3649622},
}

@article {CE25,
    AUTHOR = {Chen, Yuansi and Eldan, Ronen},
     TITLE = {Localization schemes: a framework for proving mixing bounds
              for {M}arkov chains},
   JOURNAL = {Duke Math. J.},
  FJOURNAL = {Duke Mathematical Journal},
    VOLUME = {174},
      YEAR = {2025},
    NUMBER = {8},
     PAGES = {1431--1510},
      ISSN = {0012-7094,1547-7398},
   MRCLASS = {60J22 (60H30 60K35 82M60)},
  MRNUMBER = {4916109},
MRREVIEWER = {Udrea\ P\u aun},
       DOI = {10.1215/00127094-2024-0063},
       URL = {https://doi.org/10.1215/00127094-2024-0063},
}

@article {V00,
    AUTHOR = {Vitale, Richard A.},
     TITLE = {Some comparisons for {G}aussian processes},
   JOURNAL = {Proc. Amer. Math. Soc.},
  FJOURNAL = {Proceedings of the American Mathematical Society},
    VOLUME = {128},
      YEAR = {2000},
    NUMBER = {10},
     PAGES = {3043--3046},
      ISSN = {0002-9939,1088-6826},
   MRCLASS = {60G15 (60E15)},
  MRNUMBER = {1664383},
       DOI = {10.1090/S0002-9939-00-05367-3},
       URL = {https://doi.org/10.1090/S0002-9939-00-05367-3},
}

@article {BY88,
    AUTHOR = {Bai, Z. D. and Yin, Y. Q.},
     TITLE = {Necessary and sufficient conditions for almost sure
              convergence of the largest eigenvalue of a {W}igner matrix},
   JOURNAL = {Ann. Probab.},
  FJOURNAL = {The Annals of Probability},
    VOLUME = {16},
      YEAR = {1988},
    NUMBER = {4},
     PAGES = {1729--1741},
      ISSN = {0091-1798,2168-894X},
   MRCLASS = {60F99 (15A52)},
  MRNUMBER = {958213},
MRREVIEWER = {Eric\ V.\ Slud},
       URL =
              {http://links.jstor.org/sici?sici=0091-1798(198810)16:4<1729:NASCFA>2.0.CO;2-C&origin=MSN},
}

@article{SZ81,
  title = {Dynamic Theory of the Spin-Glass Phase},
  author = {Sompolinsky, H. and Zippelius, Annette},
  journal = {Phys. Rev. Lett.},
  volume = {47},
  issue = {5},
  pages = {359--362},
  numpages = {0},
  year = {1981},
  month = {Aug},
  publisher = {American Physical Society},
  doi = {10.1103/PhysRevLett.47.359},
  url = {https://link.aps.org/doi/10.1103/PhysRevLett.47.359}
}

@book{MPV86,
author = {Mezard, M and Parisi, G and Virasoro, M},
title = {Spin Glass Theory and Beyond},
publisher = {WORLD SCIENTIFIC},
year = {1986},
doi = {10.1142/0271},
address = {},
edition   = {},
URL = {https://www.worldscientific.com/doi/abs/10.1142/0271},
eprint = {https://www.worldscientific.com/doi/pdf/10.1142/0271}
}

@article {BB19,
    AUTHOR = {Bauerschmidt, Roland and Bodineau, Thierry},
     TITLE = {A very simple proof of the {LSI} for high temperature spin
              systems},
   JOURNAL = {J. Funct. Anal.},
  FJOURNAL = {Journal of Functional Analysis},
    VOLUME = {276},
      YEAR = {2019},
    NUMBER = {8},
     PAGES = {2582--2588},
      ISSN = {0022-1236,1096-0783},
   MRCLASS = {82D30},
  MRNUMBER = {3926125},
MRREVIEWER = {Ji\ Oon\ Lee},
       DOI = {10.1016/j.jfa.2019.01.007},
       URL = {https://doi.org/10.1016/j.jfa.2019.01.007},
}

@article{EKZ22,
  author  = {Eldan, Ronen and Koehler, Frederic and Zeitouni, Ofer},
  title   = {A spectral condition for spectral gap: fast mixing in
             high-temperature {Ising} models},
  journal = {Probability Theory and Related Fields},
  year    = {2022},
  volume  = {182},
  number  = {3},
  pages   = {1035--1051},
  doi     = {10.1007/s00440-021-01085-x},
  url     = {https://doi.org/10.1007/s00440-021-01085-x},
  issn    = {1432-2064}
}

@inproceedings{AJKPV22,
author = {Anari, Nima and Jain, Vishesh and Koehler, Frederic and Pham, Huy Tuan and Vuong, Thuy-Duong},
title = {Entropic independence: optimal mixing of down-up random walks},
year = {2022},
isbn = {9781450392648},
publisher = {Association for Computing Machinery},
address = {New York, NY, USA},
url = {https://doi.org/10.1145/3519935.3520048},
doi = {10.1145/3519935.3520048},
booktitle = {Proceedings of the 54th Annual ACM SIGACT Symposium on Theory of Computing},
pages = {1418–1430},
numpages = {13},
location = {Rome, Italy},
series = {STOC 2022}
}

@INPROCEEDINGS{AMS22,
  author={Alaoui, Ahmed El and Montanari, Andrea and Sellke, Mark},
  booktitle={2022 IEEE 63rd Annual Symposium on Foundations of Computer Science (FOCS)}, 
  title={Sampling from the {S}herrington-{K}irkpatrick {G}ibbs measure via algorithmic stochastic localization}, 
  year={2022},
  volume={},
  number={},
  pages={323-334},
  doi={10.1109/FOCS54457.2022.00038}}

@article{Celetano24,
author = {Michael Celentano},
title = {{Sudakov–Fernique post-AMP, and a new proof of the local convexity of the TAP free energy}},
volume = {52},
journal = {The Annals of Probability},
number = {3},
publisher = {Institute of Mathematical Statistics},
pages = {923 -- 954},
year = {2024},
doi = {10.1214/23-AOP1675},
URL = {https://doi.org/10.1214/23-AOP1675}
}

@misc{DLSS26,
      title={Weak {P}oincar{\'e} {I}nequalities via {A}pproximate {S}tochastic {L}ocalization: {A}pplication to {S}ampling the {S}herrington-{K}irkpatrick {M}odel}, 
      author={Ewan Davies and Holden Lee and Juspreet Singh Sandhu and Jonathan Shi},
      year={2026},
      eprint={2607.08160},
      archivePrefix={arXiv},
      primaryClass={math.PR},
      url={https://arxiv.org/abs/2607.08160}, 
}

@misc{GJMPPS26,
      title={A simple proof of rapid mixing on random regular graphs beyond uniqueness}, 
      author={Andreas Göbel and Matthew Jenssen and Marcus Michelen and Marcus Pappik and Will Perkins and Leon Schiller},
      year={2026},
      eprint={2606.27545},
      archivePrefix={arXiv},
      primaryClass={math.PR},
      url={https://arxiv.org/abs/2606.27545}, 
}

@article{Opp18,
author = {Oppenheim, Izhar},
title = {Local Spectral Expansion Approach to High Dimensional Expanders Part I: Descent of Spectral Gaps},
year = {2018},
issue_date = {March     2018},
publisher = {Springer-Verlag},
address = {Berlin, Heidelberg},
volume = {59},
number = {2},
issn = {0179-5376},
url = {https://doi.org/10.1007/s00454-017-9948-x},
doi = {10.1007/s00454-017-9948-x},
journal = {Discrete Comput. Geom.},
month = mar,
pages = {293–330},
numpages = {38}
}

@article{ALO24,
author = {Anari, Nima and Liu, Kuikui and Gharan, Shayan Oveis},
title = {Spectral Independence in High-Dimensional Expanders and Applications to the Hardcore Model},
journal = {SIAM Journal on Computing},
volume = {53},
number = {6},
pages = {FOCS20-1-FOCS20-37},
year = {2024},
doi = {10.1137/20M1367696},

URL = { 
    
        https://doi.org/10.1137/20M1367696
    
    

},
eprint = { 
    
        https://doi.org/10.1137/20M1367696
    
    

}
}

@misc{Sel24,
  author       = {Mark Sellke},
  title        = {{Statistics 291: Lecture 4: Geometric and Statistical Consequences of Annealed Free Energy}},
  year         = {2024},
  month        = feb,
  note         = {Lecture notes, scribed by Zad Chin},
  url          = {https://msellke.com/courses/STAT_291/slides_notes/Stat291Lecture4.pdf},
  urldate      = {2026-09-08}
}

@article {Bor76,
    AUTHOR = {Borell, Christer},
     TITLE = {Gaussian {R}adon measures on locally convex spaces},
   JOURNAL = {Math. Scand.},
  FJOURNAL = {Mathematica Scandinavica},
    VOLUME = {38},
      YEAR = {1976},
    NUMBER = {2},
     PAGES = {265--284},
      ISSN = {0025-5521,1903-1807},
   MRCLASS = {60G15},
  MRNUMBER = {436303},
MRREVIEWER = {Raoul\ LePage},
       DOI = {10.7146/math.scand.a-11634},
       URL = {https://doi.org/10.7146/math.scand.a-11634},
}

@inproceedings {TIS76,
    AUTHOR = {Cirel'son, B. S. and Ibragimov, I. A. and Sudakov, V.
              N.},
     TITLE = {Norms of {G}aussian sample functions},
 BOOKTITLE = {Proceedings of the {T}hird {J}apan-{USSR} {S}ymposium on
              {P}robability {T}heory ({T}ashkent, 1975)},
    SERIES = {Lecture Notes in Math.},
    VOLUME = {Vol. 550},
     PAGES = {20--41},
 PUBLISHER = {Springer, Berlin-New York},
      YEAR = {1976},
   MRCLASS = {60G15},
  MRNUMBER = {458556},
MRREVIEWER = {R.\ M.\ Dudley},
}

@incollection {CCCYZ25,
    AUTHOR = {Chen, Xiaoyu and Chen, Zejia and Chen, Zongchen and Yin,
              Yitong and Zhang, Xinyuan},
     TITLE = {Rapid mixing on random regular graphs beyond uniqueness},
 BOOKTITLE = {2025 {IEEE} 66th {A}nnual {S}ymposium on {F}oundations of
              {C}omputer {S}cience---{FOCS} 2025},
     PAGES = {2170--2193},
 PUBLISHER = {IEEE Comput. Soc. Press, Los Alamitos, CA},
      YEAR = {[2025] \copyright 2025},
      ISBN = {979-8-3315-7132-0},
   MRCLASS = {68Q87},
  MRNUMBER = {5048082},
       DOI = {10.1109/FOCS63196.2025.00115},
       URL = {https://doi.org/10.1109/FOCS63196.2025.00115},
}

@book {BGL14,
    AUTHOR = {Bakry, Dominique and Gentil, Ivan and Ledoux, Michel},
     TITLE = {Analysis and geometry of {M}arkov diffusion operators},
    SERIES = {Grundlehren der mathematischen Wissenschaften [Fundamental
              Principles of Mathematical Sciences]},
    VOLUME = {348},
 PUBLISHER = {Springer, Cham},
      YEAR = {2014},
     PAGES = {xx+552},
      ISBN = {978-3-319-00226-2; 978-3-319-00227-9},
   MRCLASS = {60J25 (58J65 60J35 60J60)},
  MRNUMBER = {3155209},
MRREVIEWER = {Ming\ Liao},
       DOI = {10.1007/978-3-319-00227-9},
       URL = {https://doi.org/10.1007/978-3-319-00227-9},
}

@article{W26,
  title={Optimal {M}ixing of {G}lauber {D}ynamics for the {S}herrington-{K}irkpatrick {M}odel at {$\beta< 1/2$}},
  author={Wang, Sihan},
  journal={arXiv preprint arXiv:2608.22159},
  year={2026}
}

@article {DLP09,
    AUTHOR = {Ding, Jian and Lubetzky, Eyal and Peres, Yuval},
     TITLE = {The mixing time evolution of {G}lauber dynamics for the
              mean-field {I}sing model},
   JOURNAL = {Comm. Math. Phys.},
  FJOURNAL = {Communications in Mathematical Physics},
    VOLUME = {289},
      YEAR = {2009},
    NUMBER = {2},
     PAGES = {725--764},
      ISSN = {0010-3616,1432-0916},
   MRCLASS = {82C20 (82C26)},
  MRNUMBER = {2506768},
MRREVIEWER = {Marc\ Wouts},
       DOI = {10.1007/s00220-009-0781-9},
       URL = {https://doi.org/10.1007/s00220-009-0781-9},
}

@article {LLP10,
    AUTHOR = {Levin, David A. and Luczak, Malwina J. and Peres, Yuval},
     TITLE = {Glauber dynamics for the mean-field {I}sing model: cut-off,
              critical power law, and metastability},
   JOURNAL = {Probab. Theory Related Fields},
  FJOURNAL = {Probability Theory and Related Fields},
    VOLUME = {146},
      YEAR = {2010},
    NUMBER = {1-2},
     PAGES = {223--265},
      ISSN = {0178-8051,1432-2064},
   MRCLASS = {82C20 (60K35 82C27)},
  MRNUMBER = {2550363},
MRREVIEWER = {Elena\ A.\ Zhizhina},
       DOI = {10.1007/s00440-008-0189-z},
       URL = {https://doi.org/10.1007/s00440-008-0189-z},
}

@article{D70,
  title={Prescribing a system of random variables by conditional distributions},
  author={Dobrushin, Roland L},
  journal={Theory of Probability \& Its Applications},
  volume={15},
  number={3},
  pages={458--486},
  year={1970},
  publisher={SIAM}
}

@inproceedings{H06,
  title={A simple condition implying rapid mixing of single-site dynamics on spin systems},
  author={Hayes, Thomas P},
  booktitle={2006 47th Annual IEEE Symposium on Foundations of Computer Science (FOCS'06)},
  pages={39--46},
  year={2006},
  organization={IEEE}
}

@incollection {DS84,
    AUTHOR = {Dobrushin, R. L. and Shlosman, S. B.},
     TITLE = {Constructive criterion for the uniqueness of {G}ibbs field},
 BOOKTITLE = {Statistical physics and dynamical systems ({K}\"oszeg, 1984)},
    SERIES = {Progr. Phys.},
    VOLUME = {10},
     PAGES = {347--370},
 PUBLISHER = {Birkh\"auser Boston, Boston, MA},
      YEAR = {1985},
      ISBN = {0-8176-3300-6},
   MRCLASS = {82A05 (82A68)},
  MRNUMBER = {821306},
MRREVIEWER = {Richard\ Holley},
}

@article {DGJ09,
    AUTHOR = {Dyer, Martin and Goldberg, Leslie Ann and Jerrum, Mark},
     TITLE = {Matrix norms and rapid mixing for spin systems},
   JOURNAL = {Ann. Appl. Probab.},
  FJOURNAL = {The Annals of Applied Probability},
    VOLUME = {19},
      YEAR = {2009},
    NUMBER = {1},
     PAGES = {71--107},
      ISSN = {1050-5164,2168-8737},
   MRCLASS = {60J10 (05C15 15A60 68Q25 82B20)},
  MRNUMBER = {2498672},
MRREVIEWER = {Alessandra\ Bianchi},
       DOI = {10.1214/08-AAP532},
       URL = {https://doi.org/10.1214/08-AAP532},
}

@article{LO25,
  title={Trickle-down Theorems via C-Lorentzian Polynomials II: Pairwise Spectral Influence and Improved Dobrushin's Condition},
  author={Leake, Jonathan and Gharan, Shayan Oveis},
  journal={arXiv preprint arXiv:2510.06549},
  year={2025}
}

@article {DGJ08,
    AUTHOR = {Dyer, Martin and Goldberg, Leslie Ann and Jerrum, Mark},
     TITLE = {Dobrushin conditions and systematic scan},
   JOURNAL = {Combin. Probab. Comput.},
  FJOURNAL = {Combinatorics, Probability and Computing},
    VOLUME = {17},
      YEAR = {2008},
    NUMBER = {6},
     PAGES = {761--779},
      ISSN = {0963-5483,1469-2163},
   MRCLASS = {60K40 (82C05)},
  MRNUMBER = {2463409},
MRREVIEWER = {Lee-Peng\ Teo},
       DOI = {10.1017/S0963548308009437},
       URL = {https://doi.org/10.1017/S0963548308009437},
}

@article {BPCPSDV22,
    AUTHOR = {Blanca, Antonio and Caputo, Pietro and Chen, Zongchen and
              Parisi, Daniel and \v Stefankovi\v c, Daniel and Vigoda, Eric},
     TITLE = {On mixing of {M}arkov chains: coupling, spectral independence,
              and entropy factorization},
   JOURNAL = {Electron. J. Probab.},
  FJOURNAL = {Electronic Journal of Probability},
    VOLUME = {27},
      YEAR = {2022},
     PAGES = {Paper No. 142, 42},
      ISSN = {1083-6489},
   MRCLASS = {60J10 (82B20)},
  MRNUMBER = {4505379},
MRREVIEWER = {Robert\ Stelzer},
       DOI = {10.1214/22-ejp867},
       URL = {https://doi.org/10.1214/22-ejp867},
}
\end{document}